\documentclass[12 pt]{amsart}
\usepackage{amssymb,amsfonts,amsthm}
\usepackage{graphicx}
\usepackage{bm}
\newtheorem{thm}{Theorem}[section]
\newtheorem{cor}[thm]{Corollary}
\newtheorem{lem}[thm]{Lemma}
\newtheorem{pro}[thm]{Proposition}

\theoremstyle{definition}
\newtheorem{dfn}[thm]{Definition}
\theoremstyle{remark}

\numberwithin{equation}{section}
\begin{document}
\title[Fourier's law and mesoscopic limit equation]{From deterministic dynamics to
  thermodynamic laws II: Fourier's law and mesoscopic limit equation}
\author{Yao Li}
\address{Yao Li: Department of Mathematics and Statistics, University
  of Massachusetts Amherst, Amherst, MA, 01002, USA}
\email{yaoli@math.umass.edu}

\begin{abstract}
  This paper considers the mesoscopic limit of a stochastic energy
  exchange model that is numerically derived from deterministic
  dynamics. The law of large numbers and the central limit theorems
  are proved. We show that the limit of the stochastic
  energy exchange model is a discrete heat equation that satisfies Fourier's
  law. In addition, when the system size (number of particles) is
  large, the stochastic energy exchange is approximated by a
  stochastic differential equation, called the mesoscopic limit
  equation. 
\end{abstract}

\thanks{The author is partially supported by NSF DMS-1813246.}
\maketitle
\section{Introduction}
Fourier's law is an empirical law describing the relationship between the thermal conductivity
and the temperature profile. In 1822, Fourier concluded that ``the
heat flux resulting from thermal conduction is proportional to the
magnitude of the temperature gradient and opposite to it in sign''
\cite{fourier1822theorie}. The well-known heat equation is derived based on Fourier's
law. However, the rigorous derivation of Fourier's law from
microscopic Hamiltonian mechanics remains to be a challenge to
mathematicians and physicist \cite{bonetto2000fourier}. This challenge mainly comes from our
limited mathematical understanding to nonequilibrium statistical
mechanics. After the foundations of statistical
mechanics were established by Boltzmann, Gibbs, and Maxwell more than
a century ago, many things about nonequilibrium steady state (NESS)
remains unclear, especially the dependency of key quantities on the
system size $N$. 

There have been several studies that aim to derive Fourier's law from the
first principle. A large class of models \cite{rey2001exponential,
  rey2002fluctuations, eckmann1999non, eckmann1999entropy, eckmann2000non} use anharmonic
chains to describe heat conduction in insulating crystals. The ergodicity
(existence, uniqueness, and the speed of convergence) of nonequilibrium steady states for some (but not all) of anharmonic
chains can be rigorously proved \cite{rey2001exponential, rey2000asymptotic}. Entropy production rate
can also be studied in some cases \cite{rey2002fluctuations,
  ruelle1997entropy, ruelle1996positivity, bonetto2009heat, eckmann1999entropy}. Also, the limiting
dynamics of energy profiles of some weakly interacting Hamiltonian system
follows Ginzburg-Landau dynamics, whose scaling limit is a nonlinear
heat equation \cite{dolgopyat2011energy, liverani2011toward}. But in general, Fourier's law can only be
proved for some simple Hamiltonian models and energy exchange
models \cite{bernardin2005fourier, kipnis1982heat}. Other studies consider dynamical billiards
systems, which largely resembles the heat conduction in ideal gas. Rigorous results beyond ergodicity is extremely difficult
when a system involves multiple interacting particles
\cite{simanyi1999hard, simanyi2003proof, bunimovich1992ergodic}. But
many non-rigorous results are available. For example, many recent
studies \cite{li2013existence, grigo2012mixing, sasada2015spectral} consider the Markov energy exchange models obtained from non-rigorous derivations
in \cite{gaspard2008heat, gaspard2008heat2,
  gaspard2008derivation}. Also see \cite{lepri2003thermal} for a
review of many numerical and analytical results.

The aim of this series of paper is to derive macroscopic thermodynamic
laws, including Fourier's law, from deterministic billiards-like
models. As stated above, a fully rigorous derivation is extremely
difficult due to the limited mathematical understanding to billiards
systems with multiple interacting particles. Hence the philosophy of
this series is to use as much rigorous study as possible, and
connecting gaps between pieces of rigorous works by numerical
results. The subject of this study is a dynamical system that models
heat conduction in gas. Consider a long and thin 2D billiard table that is
connected with two heat baths with different temperatures. Many
disk-shaped moving particles are placed in the tube. Particles move
and interact freely through elastic collisions. When a particle hits
the heat bath, it receive a random force whose statistics depends on
the boundary temperature. Needless to say, this is not a
mathematically tractable problem. We lose control of a particle once
it moves into the tube. 

In \cite{li2018billiards}, we impose a localization to this billiard-like model
by adding a series of barriers into the tube. This divides a tube to a
chain of cells. Particles can collide through opennings on the barrier
but can not pass the barrier. The motivation is that
the mean free path of realistic gas particles is as short as 68 nm at
ambient pressure \cite{jennings1988mean}. Similar idea of localization
is also reported in \cite{caprini2017fourier}. Then we use numerical simulation to
study the statistics of energy exchanges between cells. Because of the
localization, energy exchange can only be made through ``effective
collisions'', which means collisions between two particles from
adjacent cells through the opening on the barrier. The time
distribution of effective collisions and the rule of energy exchange
during an effective collision are studied. A stochastic energy
exchange model is then obtained in \cite{li2018billiards}. Additional numerical
simulation shows that this stochastic energy exchange model preserves
the key asymptotical dynamics of the original billiards-like model.

In this paper, we continue to work on the mesoscopic limit of the stochastic
energy exchange model derived in \cite{li2018billiards} and further studied in
\cite{li2018polynomial}. Still consider the
ideal gas at ambient pressure. If the size of a cell is at the same
scale of the mean
free path, then a cell should contain $10^{4} \sim 10^{5}$
particles. Therefore, we should work on the stochastic energy exchange
model with a large number of particles in each cell. In this senario, each energy
exchange only changes a small proportion of the total cell energy. To
maintain the thermal conductivity unchanged, some geometric rescaling
and time rescaling
is necessary. When the number of particles per cell increases, the
size and mass of each particle must decrease correspondingly. Then we need
to rescale the time if necessary, such that particles can not
pass these openings, but the order of magnitude of the mean heat flux can be
preserved. Let $M$ be the number of particles per cell. The goal of
the rescaling is to make the number of energy exchange per unit
time $O(M)$, and the mean heat flux $O(1)$. 

We work on the stochastic energy exchange model after the
geometric and time rescaling. The rule of energy redistributions still follow
from what we have obtained in \cite{li2018billiards}. The resultant stochastic energy
exchange model resemble a slow-fast dynamical system when there are
many particles in each cell. Small energy
exchanges occur with high frequency. Each energy exchange can be
described by a function of the current energy configuration and a few
i.i.d. random variables. This motivates us to study the law of large
numbers and the central limit theorem when $M$ approaches to
infinity. We call it the mesoscopic limit, because the observable
under consideration is now the total energy of $10^{4} \sim 10^{5}$
particles. Our calculation reveals that the mesoscopic dynamics of the
stochastic energy exchange model mimics
the Landau-Ginzburg dynamics, which appears in the scaling limit of a
number of Hamiltonian systems. 

The technique used in this paper, namely the martingale problem, is
classical. It was proposed in 1970s and successfully used
to study the scaling limit of chemical reaction systems and slow-fast
hyperbolic dynamical 
systems \cite{de2015martingale, anderson2011continuous,
  stroock2007multidimensional}. We first use the result in
\cite{ethier2009markov} to show the tightness of a sequence of random
processes. Then the limit is given by the uniqueness of the solution
to the martingale problem. In some estimations, it is particularly important to
``prescribe the randomness'' to the energy exchange model. This allows
us to ``decouple'' dependent variables after some relaxation. After
decoupling, we can work on independent random variables. We remark
that these
techniques has been applied in our earlier papers
\cite{li2014nonequilibrium, li2018polynomial}. We
set up two martingale problems to prove the law of large numbers and
the central limit theorem respectively.

The law of large number shows that at the infinite-particle limit, the
stochastic energy exchange model converges to a nonlinear discrete heat equation. In
addition, this equation admits a stable equilibrium. The energy flux
starting from this stable equilibrium can be explicitly given. Hence Fourier's law is
easily derived from the equilibrium of this discrete heat
equation. At ambient pressure, $M$
is only $10^{4}$ to $10^{5}$. Therefore, random fluctuations, which is
in the magnitude of $O(M^{-1/2})$, should not be neglected. This motivates
us to further study the central limit theorem.

The central limit theorem shows that the rescaled difference between
the stochastic energy exchange model and the nonlinear discrete heat
equation is given by a timely dependent stochastic differential
equation. Combine estimates from the law of large numbers and the central
limit theorem. Some easy calculations show that the stochastic energy
exchange model is then approximated by a stochastic differential
equation with $O(M^{-1/2})$ random perturbation term. We call this
stochastic differential equation the mesoscopic limit equation. As
will be discussed in the conclusion, the nonequilibrium steady state of this
mesoscopic limit equation can be explicitly approximated by a WKB
expansions. As a result, many properties, including the long range
correlations like the one given in \cite{spohn1983long}, entropy production rates,
and fluctuation-dissipation theorems can be proved by
working on this mesoscopic limit equation. Fourier's law of the NESS
of the stochastic energy exchange model, which is a stronger result
than the Fourier's law proved in this paper, can also be proved. We decide to put these
results into our subsequent work. 

The organization of this paper is as follows. In Section 2, we
review the main result of \cite{li2018billiards}, introduce the model setting, and describe the stochastic energy model
under the geometric rescaling.  Section 3 gives the main
result. The main strategy of proof is described in Section 4. The law of
large numbers and Fourier's law are proved in Section 5. Section 6 is
about the central limit theorem and the mesoscopic limit
equation. Section 7 is the conclusion.

\section{From billiard dynamics to stochastic energy exchange model}
\subsection{Billiards model with time rescaling}
Consider an 1D chain of $N$ billiard tables (see Figure \ref{fig1}) in $\mathbb{R}^{2}$ that
are connected through nearest neighbors, denoted by $\Omega_{1}, \cdots,
\Omega_{N}$. We assume each table is a subset of $\mathbb{R}^{2}$
whose boundary is formed by finitely many piecewise $C^{3}$ curves
that are either flat or convex inward. This assumption makes the billiard system
chaotic. Then we place $M$ disk-shaped particles into each cell. The
radius of each particle is $RM^{-1/2}$, and the mass of each particle
is $2M^{-1}$. Hence the total area of particles equals $\pi R^{2}$ and the total
mass of particles equals $2$. In addition, a barrier with a hole is placed
between each adjacent pair of cells. The size of the hole is 
$2(1-\epsilon)RM^{-1/2}$ with $\epsilon \ll 1$, so that particles can
not pass the hole. 

\begin{figure}[htbp]
\centerline{\includegraphics[width = \linewidth]{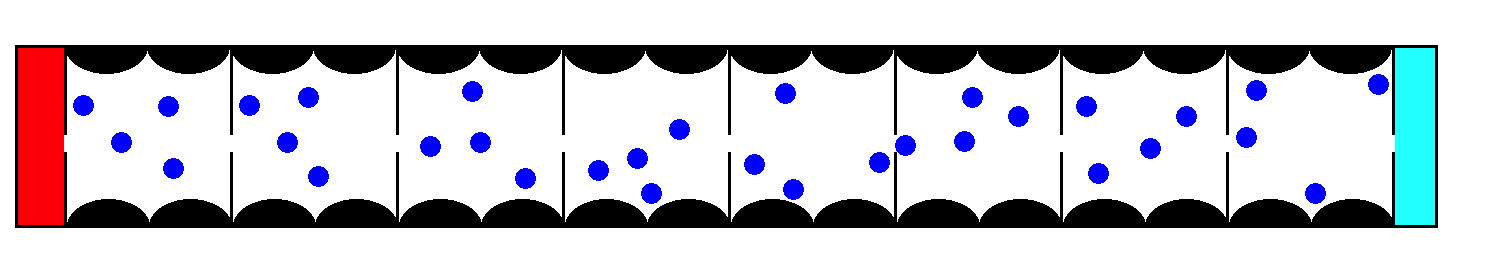}}
\caption{An 1D chain of billiard tables connected with two heat
  baths. $M = 4$ particles are ``trapped'' in each cell. A barrier
  with a hole is placed between adjacent cells, such that particles
  can collide through the hole, but cannot pass it. }
\label{fig1}
\end{figure}

Particles can move freely until colliding with the cell
boundary (including the barrier) or other particles. We assume the
following for this billiard system.
\begin{itemize}
  \item A particle is trapped by barriers and will never leave its
    cell.
\item Particles from neighbor cells can collide through holes on the
  barrier.
\item All collisions are elastic. Particles do not rotate.
\item The billiard system in each cell is chaotic.
\item $R$ is small enough such that particles would not get stuck. 
\item $R$ is small enough such that particles can be completely out of
  reach by their neighbors. 
\end{itemize}

It remains to prescribe the boundary condition. We assume that this
chain is coupled with two heat baths through the left and right
cells. The heat bath is a billiard table with the same geometric
configuration but randomly chosen total kinetic energy. After a
collision between a heat bath particle and a ``regular'' particle, a
random total energy $E_{L}$ (resp. $E_{R}$) is chosen for the left
(resp. right) heat bath from the exponential distribution with mean
$T_{L}$ (resp. $T_{R}$). Then all heat bath particles are
redistributed such that their positions and velocities satisfy the
conditional Liouville measure (conditioning on the conservation of
total energy). The system evolves deterministically between
redistributions of heat bath particles. 

The first paper in this series \cite{li2018billiards} numerically shows the
following results.
\begin{itemize}
  \item The time between two consecutive collisions through the
    barrier is exponentially distributed with a rate that can be
    approximated by $\min\{E_{1}, E_{2}\}$ if $\min \{E_{1} , E_{2}\}
    \ll 1$, where $E_{1}$ and $E_{2}$ represent the
    total energy in two cells respectively.
\item The energy carried by the particle that participates a collision
  through the barrier can be approximated by a Beta distribution with
  parameters $(1, M-1)$.
\item The energy redistribution during a collision can be approximated
  by a uniform random redistribution. 
\end{itemize} 
None of these approximation is precise. But further studies in
\cite{li2018billiards} confirms that these approximations preserve both the
asymptotic dynamics and the scaling of the thermal conductivity. 

One thing not studied in \cite{li2018billiards} is the asymptotic dependence of
collision rate on $M$. Heuristically, when $M$ is large, the mean
energy carried by each particle is only $O(M^{-1})$. In order to model
the heat conduction, we need to rescale the time to make $O(M)$
collisions per unit time. In this paper, we consider the problem at
two different time scales. Let $\phi(M)$ be the number of collisions
per unit time depending on $M$. The time rescaling $t \rightarrow
t/\phi(M)$ gives the slow scale problem, at which only $O(1)$ collision
through the barrier occurs per unit time. The time rescaling $t
\rightarrow M \phi(M)^{-1}t$ gives the fast scale problem, at which
the collision rate is $O(M)$. Our
fundamental goal is to study the limit laws of the fast scale
problem. But the slow-scale problem makes many calculations and
explanations easier. Then the time distribution between two consecutive
collisions be an exponential distribution with rate $f(E_{1}, E_{2})$ (resp. $Mf(E_{1}, E_{2})$) for the
slow (resp. fast) scale problem, where $f$ is a rate function, and
$E_{1}, E_{2}$ are the total energy stored in corresponding
cells. Note that in this paper we consider a generic rate
function $f(E_{1}, E_{2})$ that satisfies a few mild assumptions. If one take the rare collision limit first, then rescale
the time back, as did in \cite{gaspard2008heat, gaspard2008heat2}, the resultant rate function may
be different. 

The explicit formula of $\phi(M)$ is not straightforward. We
demonstrate a billiard system with two cells as an example. The total
area of particles in each cell equals to $\pi$. The geometric
configuration is shown in Figure \ref{fig2} top. Then we record the
number of collisions at the barrier and the energy of particles
that collide at the barrier for $M = 1, 2 \cdots, 50$. In Figure
\ref{fig2} bottom left, we compute the frequency of collisions through
the barrier, which is the time rescaling function $\phi(M)$ needed for the
slow-scale system. The bottom right panel of Figure \ref{fig2} shows
$M$ times the mean energy of particles that collide through the barrier. Although some
bias occurs when there are only a few particles, we can see that the scaling
of mean energy of a particle that participates a collision
is stabilized at $O(M^{-1})$. 
\begin{figure}[htbp]
\centerline{\includegraphics[width = \linewidth]{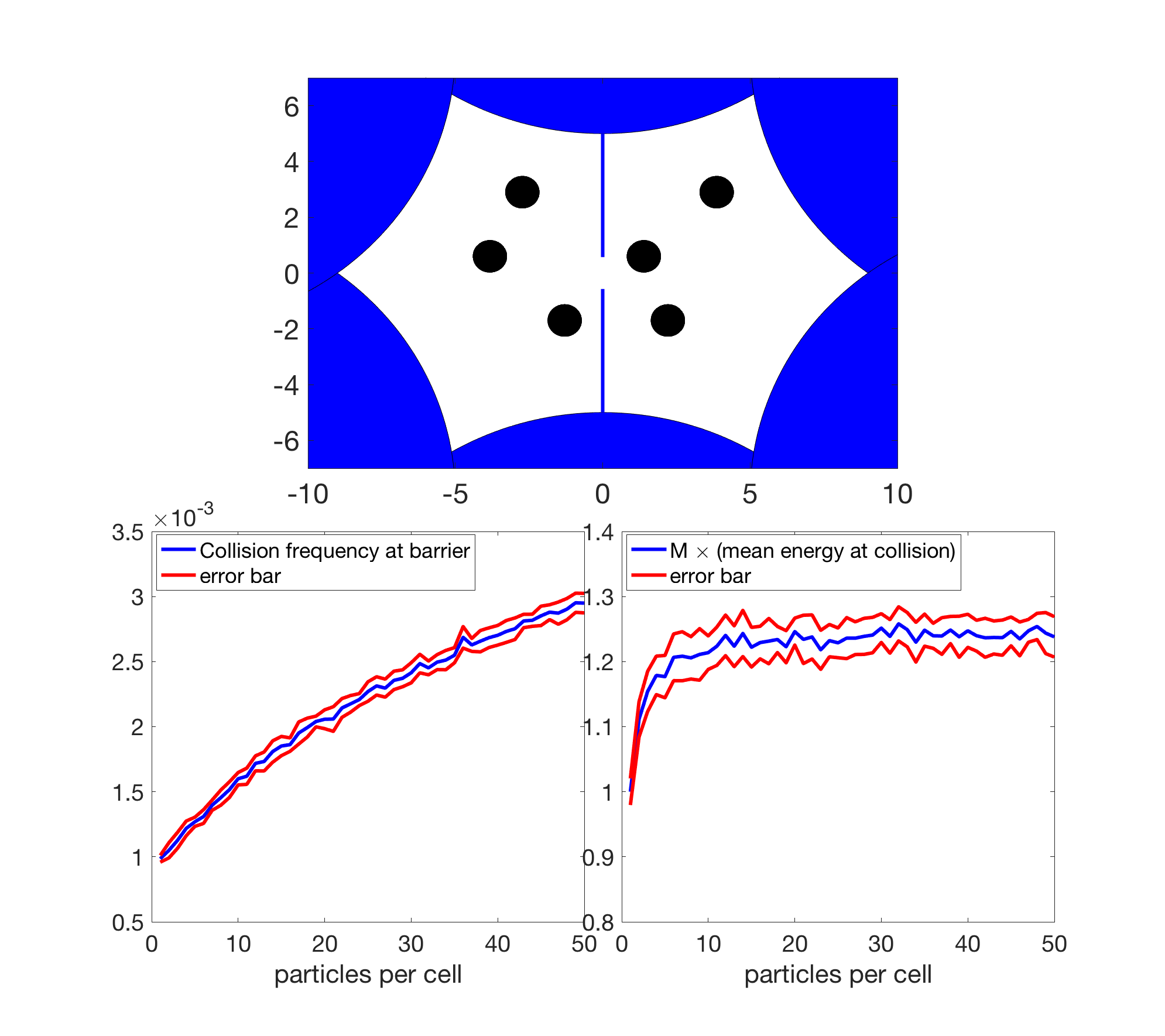}}
\caption{Top: An example of two cells with a barrier and no heat
  bath. Each cell has $3$ particles inside of it. Bottom left:
  Frequency of collision through the barrier. Red plot represents the
  error bar with one standard deviation. Bottom right: $M$ times
  the mean energy of particles that participate collisions through
  the barrier. Red plot represents the
  error bar with one standard deviation. }
\label{fig2}
\end{figure}

\subsection{Stochastic energy exchange model}
After rescaling the time by $\phi(M)$, we obtain the following
slow-scale stochastic energy
exchange model. Consider a chain of $N$ cells $1, 2, \cdots, N$ connected to two heat
baths with temperatures $T_{L}$ and $T_{R}$ respectively. Each cell contains a certain amount of
energy, denoted by $E_{1}, \cdots, E_{N}$. Let $M$ be a model
parameter that corresponds to the number of moving particles of the
original kinetic model. The rule of energy exchange is as follows. 

Assume there exists a rate function $f(E_{1}, E_{2})$ that satisfies
the following three assumptions:
\begin{itemize}
  \item[(a)] $f$ is $C^{1}$ continuous and strictly positive for
    $E_{1}, E_{2} > 0$.
\item[(b)]  $f$ is non-decreasing with respect to both $E_{1}$ and
  $E_{2}$.
\item[(c)] There exists a constant $K < \infty$ such that $f(E_{1},
  E_{2}) < K$ uniformly.
\end{itemize}
The first two assumptions are heuristic. The energy exchange rate must
be positive and continuous. Higher cell energy must have higher energy
exchange rate. The third assumption is technical. In
\cite{li2018billiards} we have showed that the stochastic energy
exchange model admits an invariant probability measure. Hence the
probability of having very high cell energy is low. Assuming all
clock rates being below a large constant will not significantly change
the dynamics. On the other hand, ``overheating'' Poisson clocks will
cause many technical troubles. The aim of this paper is to prove the limit
law. Without assumption (c), this paper will be significantly
distracted the effort of controlling ``overheating'' Poisson clocks.

An exponential clock with rate $f(E_{i}, E_{i+1})$ is associated to a
pair of cells $i$ and $i+1$. All exponential clocks are mutually
independent. When the clock rings, cells $i$ and $i+1$
exchange their energy such that 
$$
  E_{i}' = E_{i} - B_{1}E_{i} + p(B_{1}E_{i} + B_{2}E_{i+1})
$$
$$
  E_{i+1}' = E_{i+1} - B_{2}E_{i+1} + (1 - p)(B_{1}E_{i} +
  B_{2}E_{i+1}) \,,
$$
where $E_{i}'$ and $E_{i+1}'$ denote post-exchange energy, $p$ is a
uniform random variable on $(0, 1)$, and random variables $B_{1}, B_{2}$ are two
independent Beta random variables with parameters $1, M-1$. In other
words, each cell contribute a small proportion of its energy for the
redistribution. The rule of energy exchange with boundary is
similar. Two additional exponential clocks are associated to two ends
of the chain with rates $f(T_{L}, E_{1})$ and $f(E_{N}, T_{R})$. When
the clock rings, rules of update are
$$
  E_{1}' = E_{1} - B_{1}E_{1} + p(B_{1}E_{1} + B_{2}Exp(T_{L}))
$$
and 
$$
  E_{N}' = E_{N} - B_{1}E_{N} + p(B_{1}E_{N} + B_{2}Exp(T_{R})) \,,
$$
where $Exp(\lambda)$ means an exponential random variable with mean
$\lambda$, $B_{1}, B_{2}$, and $p$ are same as before. For the sake of
consistency, sometimes we use notations $E_{0} = T_{L}$ and $E_{N+1} =
T_{R}$. In addition, we denote the exponential clock between $E_{i}$
and $E_{i+1}$ by ``clock $i$''.

It is easy to see that this stochastic energy exchange model gives a
Markov chain $\Phi^{M}_{t}$ on $\mathbb{R}^{N}$, where $M$ is a model
parameter. For the sake of being consistent, we still denote the
$i$-th entry of $\Phi^{M}_{t}$ by $E_{i}(t)$ if it leads to no
confusion. 

\subsection{An alternative description}
We provide the following alternative description of $\Phi^{M}_{t}$ that
fits the calculation in this paper better. Obviously $\Phi^{M}_{t}$ is a Markov jump
process in $\mathbb{R}^{N}$. Assume $0 = t_{0} < t_{1} < \cdots$ are
jump time of $\Phi^{M}_{t}$. Then $t_{i+1} - t_{i}$ has exponential
distribution with rate 
$$
  R(\Phi^{M}_{t_{i}}) = \sum_{i = 1}^{N-1} f(E_{i}(t_{i}), E_{i+1}(t_{i})) + f(T_{L},
  E_{1}(t_{i})) + f(E_{N}(t_{i}), T_{R})  = \sum_{i = 0}^{N} f(E_{i}(t_{i}), E_{i+1}(t_{i}))\,. 
$$
We ``prescribe'' the randomness such that
$$
  t_{i+1} - t_{i} = - R( \Phi^{M}_{t_{i}})^{-1}\log(1 - q_{i}) \,,
$$
where $\{q_{i}\}$ is an i.i.d. sequence of uniform random variables on
$(0, 1)$. 

In addition, we have
$$
  \Phi^{M}_{t_{i}+} = \Phi^{M}_{t_{i}} + X^{M}_{i} \,,
$$
where $X^{M}_{i}$ is a random variable that depends on
$\Phi^{M}_{t_{i}}$ in a way that 
$$
  X^{M}_{i} = \zeta(\Phi^{M}_{t_{i}}, \omega^{M}_{i}) \,,
$$
where $\omega^{M}_{i} = (p^{(i)}_{1}, p^{(i)}_{2}, p^{(i)}_{3},
B^{(i)}_{1}, B^{(i)}_{2})$ is a sequence of i.i.d. random
vectors. $\omega^{M}_{i}$ has five entries, among which $p^{(i)}_{1}, p^{(i)}_{2}, p^{(i)}_{3}$ are three i.i.d. uniform random variables on $(0, 1)$, and
$B^{(i)}_{1}, B^{(i)}_{2}$ are two i.i.d. Beta random variables with
parameters $(1, M-1)$. 

The definition of the function $\zeta$ is as follows. $p^{(i)}_{1}$ is used to select
exponential clocks. For $0 \leq k \leq N$, clock $k$ is chosen if 
$$
  \frac{1}{R( \Phi^{M}_{t_{i}})}\sum_{j = 0}^{k}f(E_{j}(t_{i}), E_{j+1}(t_{i})) \leq p^{(i)}_{1} < \frac{1}{R( \Phi^{M}_{t_{i}})}\sum_{j = 0}^{k+1}
  f(E_{j}(t_{i}), E_{j+1}(t_{i})) \,.
$$
$p^{(i)}_{2}$ is used to choose the heat bath energy if and only if clock
$0$ or clock $N$ is chosen. $p^{(i)}_{3}$ determines the energy
redistribution. And $B_{1}^{(i)}$ and $B_{2}^{(i)}$ are two
Beta random variables involved in the energy exchange event. More
precisely, we have
\begin{equation}
  \label{exchange}
\zeta( \Phi^{M}_{t_{i}}, \omega^{M}_{i}) = \left \{ 
\begin{array}[tb]{lll}
 -J_{k}\mathbf{e}_{k}  + J_{k}\mathbf{e}_{k+1} & \mbox{ if }& \mbox{
                                                              clock }
                                                              k \mbox{
                                                              is
                                                              chosen}\,, 1 \leq k \leq N-1
  \\

J_{0} \mathbf{e}_{1}    &\mbox{ if }& \mbox{ clock } 0 \mbox{ is chosen,}
  \\
-J_{N} \mathbf{e}_{N}&\mbox{ if }&\mbox{ clock } N \mbox{ is chosen,}
\end{array}
\right .
\end{equation}
where $\mathbf{e}_{k}$ is the $k$-th vector of the standard basis, and
$$
  J_{k} = \left \{ 
\begin{array}[tb]{lll}
(1-p_{3}^{(i)})B^{(i)}_{1}E_{k}(t_{i}) - p^{(i)}_{3}B_{2}^{(i)}E_{k+1}(t_{i}) &
\mbox{ if } & 1 \leq k \leq N-1\\
-(1-p_{3}^{(i)})B^{(i)}_{1}T_{L}\log(1 - p^{(i)}_{2}) -
  p^{(i)}_{3}B_{2}^{(i)}E_{1}(t_{i}) &\mbox{ if }&k = 1 \\
(1-p_{3}^{(i)})B^{(i)}_{1}E_{N}(t_{i}) +
  p^{(i)}_{3}B_{2}^{(i)}T_{R}\log(1 - p^{(i)}_{2}) & \mbox{ if }& k =
                                                                  N 
\end{array}
\right.
$$
is the net flux from cell $k$ to cell $k+1$. 

This alternative description is less straightforward. But it
``prescribes'' all randomness in this stochastic energy process. We
will need this soon in our calculations.

\section{Main Result}
$\Phi^{M}(t)$ is a Markov jump process at the slow scale, at which the
energy flux is $O(M^{-1})$ for increasing $M$. The main result of
this paper is about the limit law of the fast scale problem with
$O(1)$ energy flux. To make the limit law work, we consider the
following process $\Theta^{M}(t)$ at the fast scale with
$$
  \Theta^{M}(t_{i}/M) =  \Phi^{M}(t_{i}) \quad \mbox{ for } t_{i} \leq
  t < t_{i+1}
$$
where $0 = t_{0} < t_{1}< \cdots< t_{n} < \cdots$ are energy exchange
times for $\Phi^{M}(t)$. Sample paths of $\Theta^{M}(t)$ are right
continuous with left limits. This makes $\Theta^{M}(t) \in D([0, T])$, the
Skorokhod space on $[0, T]$. 

The first result is about the law of large number of $D[0, T]$. 

{\bf Theorem 1}
{\it
For any finite $T > 0$, 
$$
  \lim_{M \rightarrow \infty}\Theta^{M}(t) = \bar{\Theta}(t)
$$
almost surely for any $t \in [0, T]$, where $\bar{\Theta}(t)$ solves
the ordinary differential equation
\begin{equation}
  \label{thm1ode}
\frac{\mathrm{d}}{\mathrm{d}t} \bar{\Theta}(t) = F( \bar{\Theta}(t)) \,,
\end{equation}
where
$$
F_{i}( \bar{\Theta}(t)) =
\frac{1}{2}f(\bar{\Theta}_{i-1}(t),
\bar{\Theta}_{i}(t))(\bar{\Theta}_{i-1}(t) - \bar{\Theta}_{i}(t)) +
\frac{1}{2}f(\bar{\Theta}_{i}(t),
\bar{\Theta}_{i+1}(t))(\bar{\Theta}_{i+1}(t) - \bar{\Theta}_{i}(t))   
$$ 
for $i = 1,\cdots, N$. Here we use the boundary condition
$\bar{\Theta}_{0} = T_{L}$ and $\bar{\Theta}_{N+1} = T_{R}$. 
}
 
\medskip

The Fourier's law with respect to $\bar{\Theta}(t)$ is
straightforward. We have the following proposition.

\medskip

{\bf Proposition 2}
{\it
The flow determined by equation \eqref{thm1ode} admits a stable equilibrium $E^{*}$. Let
$\kappa$ be the expected energy flux starting from $E^{*}$ (defined in
equation \eqref{kappa}), we have $\kappa = \frac{1}{2}f(T_{L}, T_{L})
+ O(T_{R} - T_{L})$ if $|T_{R} - T_{L}| \ll 1$. 
}

\medskip

Let 
$$
  \Gamma^{M}(t) = \sqrt{M}(\Theta^{M}(t) - \bar{\Theta}(t)) \,.
$$
The following theorem gives the central limit theorem for
$\Theta^{M}(t)$. 

\medskip

{\bf Theorem 3}
{\it
For any finite $T > 0$, 
$$
  \lim_{M\rightarrow \infty} \Gamma^{M}(t) = \bar{\Gamma}(t)
$$
almost surely for any $t \in [0, T]$, where $\bar{\Gamma}(t)$ solves
the time-dependent stochastic differential equation
\begin{eqnarray}
\label{thm2sde}
 \mathrm{d}\bar{\Gamma}(t)& = & DF( \bar{\Theta}(t)) \bar{\Gamma}(t)
                                \mathrm{d}t + H( \bar{\Theta}(t))
                                \mathrm{d} \mathbf{W}_{t} \\\nonumber
\bar{\Gamma}(0) &=& 0 \,,
\end{eqnarray}
where $DF$ is the Jacobian matrix of $F$ given in equation
\eqref{thm1ode}, 
$$
  H( \mathbf{E}) = 
\begin{bmatrix}
V_{0}(T_{L}, E_{1}) & V(E_{1}, E_{2}) &0&
0 &\cdots &\cdots \\
0&V(E_{1}, E_{2})& V(E_{2}, E_{3})
&0&\cdots &\cdots\\
\vdots &\vdots &\vdots&\vdots&\vdots&\vdots\\
0&\cdots&\cdots&0&V(E_{N-1}, E_{N})&V_{N}(E_{N}, T_{R})\\
\end{bmatrix} 
$$
is an $N \times (N+1)$-matrix valued function on $\mathbb{R}^{N}$, 
$$
  V(x_{1},x_{2}) = \sqrt{ f(x_{1},x_{2})\left (\frac{2}{3}x_{1}^{2} - \frac{1}{3}x_{1}x_{2} +
  \frac{2}{3}x_{2}^{2} \right )}\,,
$$
$$
  V_{0}(x_{1},x_{2}) = \sqrt{ f(x_{1},x_{2})\left (\frac{4}{3}x_{1}^{2} - \frac{1}{3}x_{1}x_{2} +
  \frac{2}{3}x_{2}^{2} \right )}\,,
$$
and
$$
  V_{N}(x_{1},x_{2}) = \sqrt{ f(x_{1},x_{2})\left (\frac{2}{3}x_{1}^{2} - \frac{1}{3}x_{1}x_{2} +
  \frac{4}{3}x_{2}^{2} \right )}\,,
$$
and $\mathrm{d} \mathbf{W}_{t}$ is the white noise in
$\mathbb{R}^{N+1}$. 
}

\medskip

Theorem 1 and Theorem 2 imply that $\Theta^{M}(t)$ is approximated by
a stochastic differential equation.  

\medskip

{\bf Proposition 4}
{\it
Let $Z_{t}$ be a stochastic differential equation satisfying 
\begin{equation}
\label{pro4sde}
  \mathrm{d}Z_{t} = F(Z_{t}) \mathrm{d}t + M^{-1/2}
  H(Z_{t})\mathrm{d}W_{t} \,.
\end{equation}
Then we have
\begin{equation}
   \mathbb{E}[ \|\Theta^{M}(t) - Z_{t}\| ]< C M^{-1}, 
\end{equation}
where $C < \infty$ is a constant that is independent of $t \in [0, T]$
and $M$.
}

\medskip

Equation \eqref{pro4sde} is called the {\it mesoscopic limit
  equation}.

\begin{proof}[Proof of main theorems]
Theorem 1 follows from Theorem \ref{LLN}. Lemma 5.9, 5.10, and 5.11 together
imply Proposition 2. Theorem 3 is proved in Section 6 in Theorem
\ref{CLT}. Proposition 4 is Corollary \ref{cor610} in Section 6.
\end{proof}

\section{Strategy of proof}
The proof of limiting laws regarding $\Theta^{M}(t)$ and
$\Gamma^{M}(t)$ can be divided into the following three steps. 

{\bf 1. Tightness.} The first step is to show that the collection of
probability measures on $D([0, T], \mathbb{R}^{N})$ generated by
$\Theta^{M}(t)$ (and $\Gamma^{M}(t)$) is tight. This means
$\Theta^{M}(t)$ (and $\Gamma^{M}(t)$) has accumulation points as $M
\rightarrow \infty$. Throughout this paper, we assume that $D([0, T],
\mathbb{R}^{N})$ is equipped with the canonical Skorokhod metric and the Borel
sigma field from it. The tightness in Skorokhod space follows from the
following well known result.

\begin{thm}[Theorem 3.8.6 of \cite{ethier2009markov}]
\label{kurtztight}
A sequence of $\mathbb{R}^{n}$-valued stochastic process $X_{n}(t), t
\in [0, \infty]$ is tight in $D([0, \infty], \mathbb{R}^{n})$ if and only if 
\begin{itemize}
  \item[(a)] For any $T > 0$, there exist $\beta > 0$ and random variable
    $\gamma_{n}(\delta, T)$ such that for $0 \leq t \leq T$, $0 \leq u
    \leq \delta$, and $0 \leq v \leq \min\{t, \delta \}$,
\begin{align*}
  \mathbb{E}[\min \left \{ q^{\beta}(X_{n}(t+u), X_{n}(t)),
  q^{\beta}(X_{n}(t), X_{n}(t-v))\right \} \,|\, \mathcal{F}^{n}_{t} ]&\leq
  \mathbb{E}[\gamma_{n}(\delta, T) \,|\, \mathcal{F}^{n}_{t}] \\
\lim_{\delta \rightarrow 0} \limsup_{n \rightarrow \infty}
  \mathbb{E}[\gamma_{n}(\delta, T)]& = 0
\end{align*}
\item[(b)]
$$
  \lim_{\delta \rightarrow 0} \limsup_{n\rightarrow \infty}
  \mathbb{E}[ q^{\beta}(X_{n}(\delta), X_{n}(0))] = 0 \,,
$$
where $q(x,y) = \min \{1, |x-y|\}$.
\end{itemize}
\end{thm}

{\bf 2. Martingale problem.} The next step is to show that any
accumulation point of $\Theta^{M}(t)$ (and $\Gamma^{M}(t)$) satisfies
a martingale problem. We have the following definition.

\begin{dfn}
Let $\mathcal{L}$ be a linear operator such that the domain of
$\mathcal{L}$ is a subset of the Banach space of all bounded Borel
measurable functions on $\mathbb{R}^{N}$. A triple $((\Omega,
\mathcal{F}, \mathbf{P}), (\mathcal{F}_{t})_{t\geq 0}, (X_{t})_{t\geq
  0})$ with $((\Omega,
\mathcal{F}, \mathbf{P}), (\mathcal{F}_{t})_{t\geq 0})$ a stochastic
basis and $X_{t}$ a $\mathcal{F}_{t}$ adapted stochastic process on
$\mathbb{R}^{N}$ is a solution of the {\it martingale problem} for
$\mathcal{L}$ if for all $f$ in the domain of $\mathcal{L}$,
$$
  f(X_{t}) - f(X_{0}) - \int_{0}^{t} ( \mathcal{L}f)(X_{s})
  \mathrm{d}s , \quad t \geq 0
$$
is a martingale with respect to $\mathcal{F}_{t}$. 
\end{dfn}

A martingale problem is said to be {\it well posed} if there exists a unique
solution $X_{t}$. Martingale problem is a very powerful tool. An obvious solution to the
martingale problem is the stochastic process whose infinitesimal
generator is $\mathcal{L}$.

{\bf 3. Uniqueness of solution to the martingale problem. } It remains to show that the maringale
problem with respect to $\Theta^{M}(t)$ (and $\Gamma^{M}(t)$) has a
unique solution. In general, let $\mathcal{L}$ be the generator of a
stochastic differential equation, then the martingale problem with
respect to $\mathcal{L}$ has a unique solution if and only if the
corresponding stochastic differential equation has a unique weak
solution. We refer \cite{stroock2007multidimensional} for further
reference regarding the uniqueness of solutions to martingale
problems. The following theorem will be used in our proof.

\begin{thm}[Theorem 10.2.2 of \cite{stroock2007multidimensional}]
\label{book1022}
If for each $T > 0$, there exists a constant $C_{T} < \infty$ such
that
$$
  \sup_{0 \leq t \leq T}\|a(t,x)\| \leq C_{T}(1 + |x|^{2}), \quad
  \mbox{for } x \in \mathbb{R}^{d}
$$
and
$$
  \sup_{0 \leq t \leq T} x \cdot b(t,x) \leq C_{T}(1 + |x|^{2}), \quad
  \mbox{for } x \in \mathbb{R}^{d} \,,
$$
then the martingale problem for generator
$$
  L_{t} = \frac{1}{2}\sum_{i,j}^{N}a_{i,j}(t, \cdot)
  \frac{\partial^{2}}{\partial x_{i} \partial x_{j}} + \sum_{i =
    1}^{N} b_{i}(t, \cdot) \frac{\partial}{\partial x_{i}}
$$
is well posed. 
\end{thm}

\section{Averaging principle and Fourier's law}
\subsection{Law of Large Numbers}
We denote the rescaled expectation of $X^{M}_{i}$ by
$$
  \mathbb{E}[X^{M}_{i}] = \frac{1}{M}\bar{\zeta}(\Phi^{M}_{t_{i}})  \,.
$$ 
It is easy to see that
$$
  \bar{\zeta} = \frac{1}{R(\Phi^{M}_{t_{i}})}
\begin{bmatrix}
\frac{1}{2}f(T_{L}, E_{1}(t_{i}))(T_{L} - E_{1}(t_{i})) +
\frac{1}{2}f(E_{1}(t_{i}), E_{2}(t_{i}))(E_{2}(t_{i}) - E_{1}(t_{i})) \\
\vdots\\
\frac{1}{2}f(E_{k}(t_{i}), E_{k+1}(t_{i}))(E_{k+1}(t_{i}) -
E_{k}(t_{i})) + \frac{1}{2}f(E_{k-1}(t_{i}),
E_{k}(t_{i}))(E_{k-1}(t_{i}) - E_{k}(t_{i})) \\
\vdots\\
 \frac{1}{2}f(E_{N}(t_{i}), E_{1}(t_{i}))(T_{R} - E_{N}(t_{i})) +
 \frac{1}{2}f(E_{N-1}(t_{i}), E_{N}(t_{i}))(E_{N-1}(t_{i}) - E_{N}(t_{i}))
\end{bmatrix}
\,.
$$

The aim of this section is to prove the law of large numbers for
$\Theta^{M}(t)$. 

\begin{thm}
\label{LLN}
For any finite $T > 0$,
$$
  \lim_{M \rightarrow \infty} \Theta^{M}(t) = \bar{\Theta}(t)
$$
almost surely, where $\bar{\Theta}(t)$ solves the ordinary
differential equation
\begin{equation}
\label{averaging}
  \frac{\mathrm{d}}{\mathrm{d}t} \bar{\Theta}(t) = R(\bar{\Theta})
  \bar{\zeta}(\bar{\Theta}) \quad , \quad \bar{\Theta}(0) = \Theta_{0}\,.
\end{equation}
\end{thm}

To prove Theorem \ref{LLN}, we first need to show that the limit of
$\Theta^{M}(t)$ solves the martingale problem. The first observation
is that two Beta random variables $B_{1}$ and $B_{2}$ are $O(M^{-1})$
small.

\begin{lem}
\label{beta}
Let $B$ be a Beta distribution with parameters $(1, M-1)$. Then for
any $0<\epsilon < 1/2$, we have
$$
  \mathbb{P}[B \geq M^{\epsilon - 1}] \leq 2e^{-M^{\epsilon}} 
$$
when $M$ is sufficiently large. 
\end{lem}
\begin{proof}
This lemma follows from straightforward calculations. The probability
density function of $B$ is $(M-1)^{-1}(1 - x)^{M-2}$. Therefore,
$$
  \mathbb{P}[B \geq M^{\epsilon -1}] = \int_{M^{\epsilon - 1}}^{1}
  (M-1)^{-1}(1 - x)^{M-2} \mathrm{d}x = (1 - M^{\epsilon - 1})^{M-1} \,.
$$
Then consider the limit
$$
  \lim_{M \rightarrow \infty} e^{M^{\epsilon}} (1 - M^{\epsilon -
    1})^{M-1} = \lim_{M \rightarrow \infty} e^{M^{\epsilon} +
    (M-1)\ln(1 - M^{\epsilon - 1})}\,.
$$
We have
$$
  \lim_{M \rightarrow \infty} M^{\epsilon} +
    (M-1)\ln(1 - M^{\epsilon - 1}) = u^{-\epsilon} + (u^{-1} -
    1)\ln (1 - u^{1-\epsilon})
$$
by changing variables $u = M^{-1}$. Take the Taylor expansion of the
logarithm, if $0 < \epsilon < 1/2$, we have
$$
  u^{-\epsilon} + (u^{-1} - 1)\ln (1 - u^{1-\epsilon}) = u^{-\epsilon}
  - u^{-\epsilon} + u^{1-\epsilon} + O(u^{1-2\epsilon}) + O(u^{2-2\epsilon})
$$
Hence 
$$
  \lim_{M \rightarrow \infty} e^{M \epsilon} (1 - M^{\epsilon -
    1})^{M-1} = 1 \,.
$$
This completes the proof.
\end{proof}

The next Lemma gives a sharper bound of $\Theta^{M}(t)$ that will be
used  in this paper. 

\begin{lem}
\label{unibound}
There exists a constant $C$ that depends on $\Theta^{M}(0)$, $K$,
$T_{L}$, $T_{R}$, and
$T$ such that
$$
  \mathbb{P}[\sup_{ t \in [0, T]}\| \Theta^{M}(t) \| > C + x ] \leq x^{-3}O(M^{-2})
$$
for all sufficiently large $M$. 
\end{lem}
\begin{proof}
Let $\mathbf{N}$ be the total number of energy exchanges on $[0, t)$. Let
$$
  I_{t} = \sum_{i = 1}^{\mathbf{N}} \mathbf{1}_{\{\mathrm{clock} \, n \, \mathrm{rings}
    \, \mathrm{at} \,  t_{i}\}}(\mathbf{1}_{\{n =
    0\}}T_{L}\log p^{(i)}_{2}B^{(i)}_{1} + \mathbf{1}_{\{n = N\}}T_{R}\log p^{(i)}_{2}B^{(i)}_{2})
$$
be the total amount of the energy influx from the boundary. Then
$$
  \sup_{ t \in [0, T]}\| \Theta^{M}(t) \| \leq T^{0} + I_{T} \,,
$$
where $T_{0} = \sum_{i = 1}^{N}E_{i}(0)$ is the initial total energy. 

Consider the worst case when all clock rates are $K$. We have 
$$
  I_{T} \leq 2\max \{ T_{L}, T_{R}\}\sum_{i = 1}^{\mathbf{N}_{0}}
  -B^{(i)}\log p^{(i)} \,,
$$
where $\mathbf{N}_{0} \sim Pois(MKT)$, $p^{(i)}$ and $B^{(i)}$ are
i.i.d. uniform 0-1 and Beta $(1, M-1)$ random variables,
respectively. 

Using Chernoff bound of Poisson tails, we have
$$
  \mathbb{P}[Pois(MKT) > 2MKT] \leq \frac{e^{-MKT}(e
    MKT)^{2MKT}}{(2MKT)^{2MKT}} = \left (\frac{\sqrt{e}}{2} \right
  )^{2MKT} \,,
$$
which is negligibly small. Hence it is sufficient to consider the tail
of 
$$
  Z_{M}:= \sum_{i = 1}^{2MKT} E^{(i)}B^{(i)} \,,
$$
where $E^{(i)} = -\log p^{(i)}$ are i.i.d. standard exponential random
variable. It is easy to see that the third central moment of
$E^{(i)}B^{(i)}$, which is $\mathbb{E}[ |E^{(i)}B^{(i)} - M^{-1}|]$,
is $O(M^{-3})$. In addition the third central moments are additive for
independent random variables. Hence the third central moment of
$Z_{M}$ is $O(M^{-2})$. Then it follows from the Chebyshev's
inequality (for higher moment) that 
$$
  \mathbb{P}[Z_{M} > 2KT + x] \leq \mathbb{P}[| Z_{M} -
  \mathbb{E}[Z_{M}]| > x] \leq x^{-3}O(M^{-2}) \,.
$$
The proof is completed by letting $C = T_{0} + 4KT\max \{ T_{L},
T_{R}\}$, as the constant $2\max \{ T_{L}, T_{R}\}$ can be absorbed
into term $O(M^{-2})$.

\end{proof}

Similar calculation gives the tightness of $\Theta^{M}(t)$ easily. 

\begin{lem}
\label{lem57}
The sequence $\{ \Theta^{M}(t) , t \in [0, T]\}$ is tight in $D([0,
T], \mathbb{R}^{N})$.
\end{lem}
\begin{proof}
For any $h > 0$, any $k = 1, \cdots, N$, and any $t \in [0, T]$, we have 
$$
 | E_{k}(t+h) - E_{k}(t)| \leq \max \left \{ (T_{0} + I_{T}) \sum_{i =
   1}^{\mathbf{N}}B^{(i)} , -\max\{T_{L}, T_{R}\}\sum_{i =
   1}^{\mathbf{N}}\log p^{(i)}B^{(i)} \right \}\,,
$$
where $\mathbf{N}$ is the Poisson random variable with rate $hKM$,
$B^{(i)}$ are i.i.d. Beta random variables with parameter $(1, M-1)$,
$p^{(i)}$ are i.i.d. uniform random variable on $(0, 1)$, and
$T_{0} + I_{T}$ is defined in the proof of Lemma \ref{unibound}. The
first term comes from ``internal'' energy exchanges, and the second
term is the boundary flux. Hence $B^{(i)}$s in the first summation are
all independent of $I_{T}$. 

Then it is easy to see that for each given $I_{T}$, the conditional
expectation of $| E_{k}(t+h) - E_{k}(t)| $ is bounded by 
$$
  \gamma(h, T) := \max \left \{ (T_{0} + I_{T})K h, -\max\{T_{L},
    T_{R}\}K h
  \right \} \,.
$$
By Lemma \ref{unibound}, the expectation of $I_{T}$ is uniformly
bounded. Hence 
$$
  \lim_{\delta \rightarrow 0}\limsup_{M \rightarrow \infty}
  \mathbb{E}[\gamma(h, T)]  = 0 \,.
$$
In addition we have
$$
  \| \Theta^{M}(t+h) - \Theta^{M}(t)\| \leq N \gamma(h, T)
$$
as $\gamma(h, T)$ is a uniform bound for all $| E_{k}(t+h) - E_{k}(t)|
$. Since $h$ can be any positive number, by Theorem \ref{kurtztight},
$\{\Theta^{M}(t), t \in [0, T] \}$ is tight in Skorokhod space.

\end{proof}

Now we are ready to set up the martingale problem.

\begin{lem}
\label{martingale}
For any test function $A \in C_{c}^{\infty}( \mathbb{R}^{N})$ and any
$t \in [0, T]$, we have
$$
  \lim_{M\rightarrow \infty} \mathbb{E}\left[ A(\Theta^{M}(t)) - A(
    \Theta^{M}(0)) - \int_{0}^{t}\nabla A( \Theta^{M}(s)) \cdot R(
    \Theta^{M}(s))\bar{\zeta}( \Theta^{M}(s)) \mathrm{d}s  \right] = 0 \,.
$$
\end{lem}
\begin{proof}
Let $\mathcal{L}^{M}$ be the infinitesimal generator of
$\Theta^{M}(t)$. It is easy to see that for any $\mathbf{E} = (E_{1},
\cdots, E_{N})$, we have
\begin{align*}
  \mathcal{L}^{M}A( \mathbf{E}) &= Mf(T_{L}, E_{1}) \int_{[0,
  1]^{3}\times \mathbb{R}^{+}} (A(E_{1} - (1-p)B_{1}E_{1} +
  pB_{2}X_{L}, \cdots, E_{N}) - A(\mathbf{E})) \\&Q_{L}(p,B_{1},B_{2},X_{L})
  \mathrm{d}p \mathrm{d}B_{1} \mathrm{d}B_{2} \mathrm{d}X_{L}\\
+&Mf(E_{N}, T_{R}) \int_{[0,
  1]^{3}\times \mathbb{R}^{+}} (A(E_{1}, \cdots, E_{N} - (1-p)B_{1}E_{N} +
  pB_{2}X_{R}) - A(\mathbf{E})) \\&Q_{R}(p,B_{1},B_{2},X_{L})
  \mathrm{d}p \mathrm{d}B_{1} \mathrm{d}B_{2} \mathrm{d}X_{R}\\
+& M\sum_{i = 1}^{N-1}f(E_{i}, E_{i+1}) \int_{[0, 1]^{3}}(A(\cdots,
   E_{i} - (1-p)B_{1}E_{i} + pB_{2}E_{i+1}, \\
& E_{i+1} - pB_{2}E_{i+1} +
   (1-p)B_{1}E_{i}, \cdots) - A( \mathbf{E})) Q(p, B_{1},B_{2})
  \mathrm{d}p \mathrm{d}B_{1} \mathrm{d}B_{2} \,,
\end{align*}
where $Q$ is the joint probability density function of two independent
Beta random variables $B_{1}, B_{2}$ with parameters $(1, M-1)$ and one uniform
random variable $p$ on $(0, 1)$, $Q_{L}$ (resp. $Q_{R}$) is the joint
probability density function of two independent
Beta random variables $B_{1}, B_{2}$ with parameters $(1, M-1)$, one
exponential random variable with mean $T_{L}$ (resp. $T_{R}$), and one uniform
random variable $p$ on $(0, 1)$. 

Since $\Theta^{M}(t)$ is a Markov process with zero probability of
hitting infinity in finite time, for any test function $A(
\mathbf{E})$, 
$$
  A(\Theta^{M}(t)) - A(\Theta^{M}(0)) - \int_{0}^{t}
  \mathcal{L}^{M}A(\Theta^{M}(s)) \mathrm{d}s
$$
is a martingale. Hence it is sufficient to show that 
$$
\lim_{M \rightarrow \infty} \int_{0}^{t} | \mathcal{L}^{M}A( \Theta^{M}(s)) - \nabla
  A( \Theta^{M}(s)) \cdot R(  \Theta^{M}(s)) \bar{\zeta}(
  \Theta^{M}(s))  | \mathrm{d}s = 0\,.
$$

By Lemma \ref{beta}, two Beta random variables are $O(M^{-1})$
small. More precisely, the probability of $B_{i}^{(1)} > M^{\epsilon - 1}$ or
$B_{i}^{(2)} > M^{\epsilon - 1}$ for some $t_{i} < t$ is
$O(M)e^{-M^{\epsilon}} $, which is smaller than any powers of $M$. 

A Taylor expansion of $A( \mathbf{E})$ gives 
\begin{align*}
  &\mathcal{L}^{M}A( \mathbf{E}) \\
= &Mf(T_{L}, E_{0})A_{E_{1}}\int_{[0,
                                  1]^{3} \times \mathbb{R}^{+}}
                                  (-(1-p)B_{1}E_{1} + p B_{2}X_{L})Q_{L}(p,B_{1},B_{2},X_{L})
  \mathrm{d}p \mathrm{d}B_{1} \mathrm{d}B_{2} \mathrm{d}X_{L} \\
&+ Mf(E_{N}, T_{R})A_{E_{N}}\int_{[0,
                                  1]^{3} \times \mathbb{R}^{+}}
                                  (-(1-p)B_{1}E_{1} + p B_{2}X_{R})Q_{R}(p,B_{1},B_{2},X_{R})
  \mathrm{d}p \mathrm{d}B_{1} \mathrm{d}B_{2} \mathrm{d}X_{R}\\
&+ \sum_{i = 1}^{N-1} \left ( Mf(E_{i}, E_{i+1})  \int_{[0,
  1]^{3}}A_{E_{i}} (-(1-p)B_{1}E_{i} + pB_{2}E_{i+1}) + \right .\\
& \quad \left . A_{E_{i+1}}(-p
  B_{2}E_{i+1} + (1-p)B_{1}E_{i}) Q(p, B_{1},B_{2}) \mathrm{d}p
  \mathrm{d}B_{1} \mathrm{d}B_{2} \right ) + M\epsilon(\mathbf{E}) O(M^{-2}) \\
=&\sum_{i = 1}^{N} [\frac{1}{2}f(E_{i-1}, E_{i})(E_{i-1} - E_{i}) +
\frac{1}{2}f(E_{i}, E_{i+1})(E_{i+1}- E_{i})] A_{E_{i}} + \epsilon(\mathbf{E})O(M^{-1})\\
=& \nabla
  A( \mathbb{E}) \cdot R( \mathbf{E}) \bar{\zeta}( \mathbf{E}) +
   \epsilon(\mathbf{E})O(M^{-1}) \,,
\end{align*}
where the Lagrange reminder $\epsilon(\mathbf{E})$ depends on second derivatives of $A$ and
is bounded in a compact set.

Therefore, we have
\begin{equation}
\label{eq5-1}
  \int_{0}^{t} | \mathcal{L}^{M}A( \Theta^{M}(s)) - \nabla
  A( \Theta^{M}(s)) \cdot R(  \Theta^{M}(s)) \bar{\zeta}(
  \Theta^{M}(s))  | \mathrm{d}s =
  \int_{0}^{t}\epsilon(\Theta^{M}(s))O(M^{-1}) \mathrm{d}s \,.
\end{equation}

By Lemma \ref{unibound}, the probability that $\sup_{ s \in
  [0, t]}\| \Theta^{M}(s) \| > C + 1 $ for some constant $C$ is
$O(M^{-2})$ small. In addition $A$ is a bounded function, hence the impact of very
  large $\Theta^{M}(s)$ is negligible. Therefore, the lemma follows
  from equation \eqref{eq5-1}.
\end{proof}

\begin{proof}[Proof of Theorem \ref{LLN}.]
The boundedness of $\Theta^{M}(0)$ is trivial. Hence it follows from Lemma \ref{lem57}, Lemma \ref{unibound}, and Theorem \ref{kurtztight} that $\{ \Theta^{M}(t)
, t \in [0, T]\}$ is tight. It is then sufficient to show that equation \eqref{averaging} is the only
solution of the martingale problem given in Lemma
\ref{martingale}. Let $\Theta(t)$ be a solution to the martingale
problem in Lemma \ref{martingale}. Applying Lemma \ref{martingale} to
the identity function, we have
\begin{eqnarray*}
&&\frac{\mathrm{d}}{\mathrm{d}t} \mathbb{E}_{\Theta_{0}}[\| \Theta(t) -
  \bar{\Theta}(t) \|^{2}]\\
 & = & \frac{\mathrm{d}}{\mathrm{d}t}
                                \mathbb{E}_{\Theta_{0}}[ \Theta(t)
                                \cdot \Theta(t)] -
                                2R(\bar{\Theta}(t))\bar{\zeta}(\bar{\Theta}(t))\cdot
  \mathbb{E}_{\Theta_{0}}[\Theta(t)]\\
&& - 2 \bar{\Theta}(t)\cdot(\frac{\mathrm{d}}{\mathrm{d}t}
   \mathbb{E}_{\Theta_{0}}[\Theta(t)]) + 
   2R(\bar{\Theta}(t))\bar{\zeta}(\bar{\Theta}(t))\cdot
   \bar{\Theta}(t) \\
&=&\mathbb{E}_{\Theta_{0}}[ \Theta(t)
                                \cdot 2R(\Theta(t))\bar{\zeta}(\Theta(t))]-
                                2R(\bar{\Theta}(t))\bar{\zeta}(\bar{\Theta}(t))\cdot
  \mathbb{E}_{\Theta_{0}}[\Theta(t)]\\
&& - 2 \bar{\Theta}(t)\cdot \mathbb{E}_{\Theta_{0}}[R(\Theta(t))\bar{\zeta}(\Theta(t))] + 2
    \mathbb{E}_{\Theta_{0}}[R(\bar{\Theta}(t))\bar{\zeta}(\bar{\Theta}(t))\cdot
   \bar{\Theta}(t)] \\
&=&2 \mathbb{E}_{\Theta_{0}}[(\Theta(t) - \bar{\Theta}(t)) \cdot
    (R(\Theta(t))\bar{\zeta}(\Theta(t)) -
    R(\bar{\Theta}(t))\bar{\zeta}(\bar{\Theta}(t))) ] \,.
\end{eqnarray*}

Since the rate function $f$ is globally bounded, some elementary calculations imply that
$$
  R(\Theta(t))\bar{\zeta}(\Theta(t)) -
    R(\bar{\Theta}(t))\bar{\zeta}(\bar{\Theta}(t)) \leq K A (
    \Theta(t) - \bar{\Theta}(t)) \,,
$$
where
$$
  A = 
\begin{bmatrix}
1& -1/2& 0 &\cdots&0\\
-1/2&1&-1/2&\cdots&0\\
\vdots&\vdots&\vdots&\vdots&\vdots\\
0&\cdots&0&-1/2&1
\end{bmatrix} \,.
$$
Hence there exists a constant $C$ such that 
$$
  \frac{\mathrm{d}}{\mathrm{d}t} \mathbb{E}_{\Theta_{0}}[\| \Theta(t) -
  \bar{\Theta}(t) \|^{2}] \leq C \mathbb{E}_{\Theta_{0}}[\| \Theta(t) -
  \bar{\Theta}(t) \|^{2}] \,.
$$
Since $\Theta(0) = \Theta_{0}$, by Gronwall's inequality, we have
$\Theta(t) = \bar{\Theta}(t)$ almost surely. This completes the solution.

\end{proof}

\subsection{Fourier's law of the limit equation}
\begin{lem}
\label{lem61}
Equation \eqref{averaging} admits a unique equilibrium $\mathbf{E}^{*} =
(E^{*}_{1}, \cdots, E^{*}_{N})$ in $\mathbb{R}^{N}_{+}$.
\end{lem}
\begin{proof}
We have
\begin{displaymath}
 0 = R(\bar{\Theta}) \bar{\zeta}(\Theta)= \begin{bmatrix}
\frac{1}{2}f(T_{L}, E_{1}^{*})(T_{L} - E_{1}^{*}) +
\frac{1}{2}f(E_{1}^{*}, E_{2}^{*})(E_{2}^{*} - E_{1}^{*}) \\
\vdots\\
\frac{1}{2}f(E_{k}^{*}, E_{k+1}^{*})(E_{k+1}^{*} -
E_{k}^{*}) + \frac{1}{2}f(E_{k-1}^{*},
E_{k}^{*})(E_{k-1}^{*} - E_{k}^{*}) \\
\vdots\\
 \frac{1}{2}f(E_{N}^{*}, E_{1}^{*})(T_{R} - E_{N}^{*}) +
 \frac{1}{2}f(E_{N-1}^{*}, E_{N}^{*})(E_{N-1}^{*} - E_{N}^{*})
\end{bmatrix}
\end{displaymath}
Therefore, we have
$$
  f(T_{L}, E_{1}^{*})(E_{1}^{*} - T_{L} ) = f(E_{1}^{*},
  E_{2}^{*})(E_{2}^{*} - E_{1}^{*}) = \cdots = f(E^{*}_{N},
  T_{R})(T_{R} - E^{*}_{N}) \,.
$$
We can use this identity to match the left and right boundary
conditions. For any $c > 0$, we can solve equation
$$
  f(T_{L}, E_{1}^{*})(E_{1}^{*} - T_{L}) = c \,.
$$
Denote the solution by $E_{1}^{*}(c)$. Since $f$ is positive, we have
$E_{1}^{*}(c) > T_{L}$. By the continuity of $f$, $E_{1}^{*}(c)$ is
continuous with respect to $c$. In addition, since
$$
  \frac{\mathrm{d}}{\mathrm{d}E^{*}_{1}}(f(T_{L}, E_{1}^{*})(E_{1}^{*}
  - T_{L})) = f_{2}(T_{L}, E_{1}^{*})(E_{1}^{*} - T_{L}) + f(T_{L},
  E_{1}^{*}) > 0 \,,
$$ 
$E_{1}^{*}(c)$ monotonically increases with $c$. Similarly, we can
solve equation
$$
  f(E_{1}^{*}(c), E^{*}_{2})(E_{2}^{*} - E_{1}^{*}(c)) = c
$$
to get $E^{*}_{2}(c)$. And $E^{*}_{2}(c)$ increases with $c$ by the
same reason as that of $E^{*}_{1}(c)$. Continue this procedure, we can
obtain $E_{3}^{*}(c), \cdots, E^{*}_{N}(c)$, and $T_{R}^{*}(c)$. The
boundary value $T_{R}^{*}(c)$ is continuous with respect to $c$ and
monotonically increasing with $c$. 

Since $T_{R}^{*}(c) = 0$ and $T_{R}^{*}(+\infty) = +\infty$, by
the intermediate value theorem, there exists a $c^{*}$ such that
$$
  T_{R}^{*}(c^{*}) = T_{R} \,.
$$
It is easy to see that $(E_{1}^{*} , \cdots, E_{N}^{*}) =
(E_{1}^{*}(c^{*}), \cdots, E_{N}^{*}(c^{*}))$ is a solution to
equation \eqref{averaging}. 
\end{proof}

\begin{lem}
\label{lem510}
Assume $\gamma = \mathrm{div}f/f$ has negative partial derivatives in
a neighborhood of $\mathbf{E}^{*}$, then the equilibrium $\mathbf{E}^{*}$ for equation
\eqref{averaging} is linearly stable for sufficiently large $N$. 
\end{lem}
\begin{proof}
Without loss of generality let $E_{0} = T_{L}$ and $E_{N} =
T_{R}$. Let ${\bm J} = \{{\bm J}_{i,j}\}_{i,j = 1}^{N}$ be the Jacobian matrix of equation
\eqref{averaging} at $\mathbf{E}^{*}$. Denote two partial derivatives
of $f$ by $f_{1}$ and $f_{2}$. We have
$$
  {\bm J}_{i,i} = f_{2}(E_{i-1}^{*}, E_{i}^{*})(E_{i-1}^{*} - E_{i}^{*}) +
  f_{1}(E_{i}^{*}, E_{i+1}^{*})(E_{i+1}^{*}-E_{i}^{*}) -
  f(E_{i-1}^{*}, E_{i}^{*}) - f(E_{i}^{*}, E_{i+1}^{*}) 
$$
for $i = 1, \cdots, N$, 
$$
  {\bm J}_{i-1,i} = f_{1}(E^{*}_{i-1},E^{*}_{i})(E^{*}_{i-1}-E^{*}_{i}) +
  f(E_{i-1}^{*}, E_{i}^{*})
$$
for $i = 2, \cdots, N$, and
$$
  {\bm J}_{i, i+1} = f_{2}(E^{*}_{i}, E^{*}_{i+1})(E^{*}_{i+1} -E_{i}^{*}) +
  f(E_{i}^{*}, E_{i+1}^{*}) 
$$
for $i = 1, 2, \cdots, N-1$. All other ${\bm J}_{i,j}$ with $|i-j|>2$ are
zero. 

We have 
\begin{eqnarray*}
\sum_{j = 0}^{N}{\bm J}_{i,j} &=& {\bm J}_{i,i-1}+{\bm J}_{i,i} + {\bm J}_{i,i+1} \\
&=& (f_{1}(E^{*}_{i-1},E^{*}_{i}) +
  f_{2}(E^{*}_{i-1},E^{*}_{i}))(E_{i-1}^{*} - E^{*}_{i}) +
  (f_{1}(E^{*}_{i}, E^{*}_{i+1}) + f_{2}(E^{*}_{i},
  E^{*}_{i+1}))(E^{*}_{i+1} -E_{i}^{*}) \\
&=& -c^{*}\frac{\mathrm{div}f}{f}(E_{i-1}^{*}, E_{i}^{*}) +
    c^{*}\frac{\mathrm{div}f}{f}(E_{i}^{*}, E_{i+1}^{*})  \\
&=& c^{*}(\gamma(E^{*}_{i},E^{*}_{i+1}) - \gamma(E^{*}_{i-1},
    E^{*}_{i}) ) \,,
\end{eqnarray*}
where $c^{*}$ is the critical value given in the proof of Lemma
\ref{lem61} such that
$$
  f(E_{i-1}^{*}, E_{i}^{*})(E^{*}_{i} - E^{*}_{i}) = c^{*}
$$
for all $i = 1, \cdots, N+1$. Therefore, 
$$
  \sum_{j = 0}^{N}{\bm J}_{i,j} < 0 \,.
$$
In addition, note that by the assumption of $f$ we have
$$
  E_{i}^{*} - E_{i+1}^{*} < \frac{T_{R} - T_{L}}{N f(T_{R}, T_{R})} \,.
$$
In addition $T_{L} < E_{1}^{*} < \cdots < E_{N}^{*} < T_{R}$ according
to the proof of Lemma \ref{lem61}. Hence $J_{i,i} < 0$ when $N$ is
sufficiently large. Therefore, $J$ is a diagonally dominant matrix. By
the Gershgorin disk theorem, all eigenvalues of $J$ has strictly
negative real parts. This completes the proof. 
\end{proof}

Let 
\begin{eqnarray}
\label{kappa}
&&\\\nonumber
\kappa &= &\frac{M}{T_{R} - T_{L}} \mathbb{E}_{\mathbf{E}^{*}} \left
            \{ f(T_{L}, E_{1})[(1-p_{1})B_{1}E_{1} - p_{1}B_{0} Exp (T_{L})]\right .\\\nonumber
&&+ \sum_{i = 1}^{N-1}f(E_{i}, E_{i+1})[-(1-p_{i})B_{i}E_{i} +
    p_{i}B_{i+1}E_{i+1}] \\\nonumber
&&+ \left . f(E_{N}, T_{R})[-(1-p_{N+1})B_{N}E_{N} +
    p_{N+1}B_{N+1}Exp(T_{R})] \right \}\nonumber
\end{eqnarray}
be the thermal conductivity of the rescaled system
$\Theta^{M}(t)$, where $B_{0}, B_{1}, \cdots, B_{N+1}$ are i.i.d. Beta random
variables with parameter $(1, M-1)$, $p_{1}, \cdots, p_{N+1}$ are
i.i.d. uniform random variables on $(0, 1)$, and $Exp(\lambda)$ means
an exponential random variable with mean $\lambda$. The following
lemma implies Fourier's law.

{\bf Remark.} It remains to check partial derivatives of $\gamma$. Since $f$ is the rate function obtained from billiards-like dynamics,
heuristically $f(E, E)$ should be proportional to $\sqrt{E}$, which
has a negative second order derivative. Consider two concrete examples
of rate functions $f_{1}(E_{1}, E_{2}) = \sqrt{E_{1}, E_{2}}$ and $f_{2}(E_{1}, E_{2}) =
\sqrt{E_{1}E_{2}/(E_{1} + E_{2})}$ that has been considered in previous
studies, where $f_{1}$ is the rate function
obtained by taking the rare interaction limit \cite{gaspard2008heat,
  gaspard2008heat2}, and $f_{2}$ satisfies with our conclusion in
\cite{li2018billiards} that $f_{2}(E_{1}, E_{2}) \approx
\sqrt{\min\{E_{1}, E_{2}\}}$ if one of $E_{1}$ and $E_{2}$ is small.  

Some elementary calculations show that
$$
  \gamma_{1} = \frac{\mathrm{div} f_{1}}{f_{1}} = \frac{1}{E_{1} + E_{2}}
$$
and
$$
  \gamma_{2} = \frac{\mathrm{div} f_{2}}{f_{2}} = \frac{E_{1}^{2} +
    E_{2}^{2}}{(E_{1} + E_{2})E_{1}E_{2}} \,.
$$
Partial derivatives of $\gamma_{1}$ are always negative. Partial
derivatives of $\gamma_{2}$ are negative if
$$
  (1+\sqrt{2})^{-1}E_{1} < E_{2} < (1+\sqrt{2})E_{1} \,.
$$
Hence when the chain is sufficiently long, $\gamma_{2}$ also satisfies
the assumption in Lemma \ref{lem510} because $E^{*}_{i} - E^{*}_{i+1}
= O(N^{-1})$.

\begin{lem}
Assume $T_{R} - T_{L} \ll 1$. Then
$$
  \kappa = \frac{1}{2}f(T_{L}, T_{L}) + O(T_{R} - T_{L}) \,.
$$ 
\end{lem}
\begin{proof}
Taking the expectation, it is easy to see that
$$
  \kappa = \frac{1}{2(T_{R} - T_{L})}\sum_{i = 0}^{N}f(E_{i}^{*},
  E_{i+1}^{*})(E_{i+1}^{*} - E_{i}^{*}) = \frac{c^{*}(N+1)}{2(T_{R} -
    T_{L})} \,.
$$
By the definition of $c^{*}$, we have
$$
  \sum_{i = 0}^{N} \frac{c^{*}}{f(E_{i}^{*}, E_{i+1}^{*}) } = T_{R} -
  T_{L} \,.
$$
By the monotonicity of $f$, we have
$$
  \frac{1}{N+1}f(T_{L}, T_{L})(T_{R} - T_{L}) < c^{*}
  < \frac{1}{N+1}f(T_{R}, T_{R})(T_{R} - T_{L}) \,.
$$
The result follows from a Taylor expansion of $f$. 
\end{proof}

\section{Central limit theorem}

Let
$$
  \Gamma^{M}(t) = \sqrt{M}(\Theta^{M}(t) - \bar{\Theta}(t) ) \,,
$$
where $\bar{\Theta}(t)$ solves equation \eqref{averaging}. 

The main result of this section is the following Theorem.

\begin{thm}
\label{thm61}
For any finite $T > 0$, 
$$
  \lim_{M\rightarrow \infty} \Gamma^{M}(t) = \bar{\Gamma}(t)
$$
almost surely, where $\bar{\Gamma}(t)$ solves the time-dependent stochastic
differential equation
\begin{eqnarray}
\label{CLT}
  \mathrm{d} \bar{\Gamma}(t) &=& D(R( \bar{\Theta}(t))
  \bar{\zeta}(\bar{\Theta}(t)) )\bar{\Gamma}(t) \mathrm{d}t + H(\bar{\Theta}(t)) \mathrm{d}\mathbf{W}_{t}\\\nonumber
\bar{\Gamma}(0) &=& 0 \,,
\end{eqnarray}
where 
$$
  H( \mathbf{E}) = 
\begin{bmatrix}
V_{0}(T_{L}, E_{1}) & V(E_{1}, E_{2}) &0&
0 &\cdots &\cdots \\
0&V(E_{1}, E_{2})& V(E_{2}, E_{3})
&0&\cdots &\cdots\\
\vdots &\vdots &\vdots&\vdots&\vdots&\vdots\\
0&\cdots&\cdots&0&V(E_{N-1}, E_{N})&V_{N}(E_{N}, T_{R})\\
\end{bmatrix} 
$$
is an $N \times (N+1)$-matrix valued function on $\mathbb{R}^{N}$, 
$$
  V(x_{1}, x_{2}) = \sqrt{ f(x_{1},x_{2})\left ( \frac{1}{4}x_{1}^{2} +
    \frac{1}{6}x_{1}x_{2} + \frac{1}{4} x_{2}^{2}\right )} \,,
$$
$$
  V_{0}(x_{1},x_{2}) = \sqrt{ f(x_{1},x_{2})\left (\frac{3}{4}x_{1}^{2} + \frac{1}{6}x_{1}x_{2} +
  \frac{1}{4}x_{2}^{2} \right )}\,,
$$
$$
  V_{N}(x_{1},x_{2}) = \sqrt{ f(x_{1},x_{2})\left (\frac{1}{4}x_{1}^{2} + \frac{1}{6}x_{1}x_{2} +
  \frac{3}{4}x_{2}^{2} \right )}\,,
$$
and $\mathrm{d} \mathbf{W}_{t}$ is the white noise in
$\mathbb{R}^{N+1}$. 
\end{thm}

The proof of Theorem \ref{CLT} is divided into the following steps. We
first prove the tightness of $\{\Gamma^{M}(t), t \in [0, T]\}$ by using Theorem
\ref{kurtztight}. Then Lemma
\ref{lem64} shows that any sequential limit of $\Gamma^{M}(t)$ solves
a martingale problem. The second order derivative term in this
martingale problem is explicitly calculated in Lemma
\ref{covariance}. Finally, Lemma \ref{lem65} shows
the uniqueness of solutions to the martingale problem described in
Lemma \ref{lem64}. Theorem \ref{CLT} follows from these lemmata. 

\begin{lem}
\label{gammatight}
The sequence $\{\Gamma^{M}(t), t \in[0, T]\}$ is tight in $D([0, T],
\mathbb{R}^{N}\}$. 
\end{lem}
\begin{proof}
By Theorem \ref{kurtztight}, it is sufficient to show that for any $t
\in [0, T]$ and any sufficiently small $h$,
$$
  \mathbb{E}[\| \Gamma^{M}(t+h) - \Gamma^{M}(t) \|^{2}] \leq C |h| 
$$
for all sufficiently large $M$, where $C$ is a constant
independent of $M$.

Without loss of generality assume $\Gamma^{M}(t) = 0$. Recall that $\Theta^{M}(t)
\rightarrow \bar{\Theta}(t)$. Let $t_{i}$ be the time of $i$-th clock
ring after $t$. For each $N$, we have
\begin{align*}
  \Gamma^{M}(t_{N}) =& \sqrt{M}\sum_{k = 0}^{N-1}\left[ \zeta(
                      \Theta^{M}(t_{k}), \omega_{k}) - (
                      \bar{\Theta}(t_{k+1}) - \bar{\Theta}(t_{k})
                      )\right]\\
=&\sqrt{M}\sum_{k = 1}^{N-1}\left [ \zeta( \Theta^{M}(t_{k}),
  \omega_{k}) - R \bar{\zeta}( \bar{\Theta}(t_{k}))(t_{k+1} -
  t_{k})\right ] + O(M^{-1/2})\\
=&\sqrt{M}\sum_{k = 1}^{N-1}\left[ \zeta(\Theta^{M}(t_{k}), \omega_{k}) -
  R\bar{\zeta}(\Theta^{M}(t_{k}))(t_{k+1} - t_{k}) \right ] \\
&+ \sum_{k =
  1}^{N-1}D(R \bar{\zeta})(
  \bar{\Theta}(t_{k}))\sqrt{M}(\Theta^{M}(t_{k}) -
  \bar{\Theta}(t_{k}))(t_{k+1} - t_{k}) \\
&+ O(M^{-1/2}) + o(M^{-3/2}\|
  \Gamma^{M}(t_{k})\|^{2} ) \\
=&\sum_{k = 0}^{N-1} \hat{\zeta}(\Theta^{M}(t_{k}), \omega_{k}) +
  \sum_{k = 0}^{N-1} D(R \bar{\zeta})(
  \bar{\Theta}(t_{k})) \Gamma^{M}(t_{k})(t_{k+1} - t_{k}) \\
&+ O(M^{-1/2}) + o(M^{-3/2}\|
  \Gamma^{M}(t_{k})\|^{2} ) \,,
\end{align*}
where
$$
  \hat{\zeta}(\Theta^{M}(t_{k}), \omega_{k}) = \sqrt{M}\left (\zeta(\Theta^{M}(t_{k}), \omega_{k}) -
  R\bar{\zeta}(\Theta^{M}(t_{k})) \mathcal{E}_{k}\right )
$$
are independent random variables with zero mean, $\mathcal{E}_{k}$ is an
exponential random variable with mean $R(\Theta^{M}(t_{k}))M$. Easy
calculation shows that
$$
  \mathbb{E}[\|\hat{\zeta}(\Theta^{M}(t_{k}), \omega_{k})\|^{2} ] =
  O(M^{-1}) \,.
$$
Denote $\hat{\zeta}(\Theta^{M}(t_{k}), \omega_{k})$ by
$\hat{\zeta}_{k}$. Since $\hat{\zeta}_{i}$ is independent of
$\hat{\zeta}_{j}$ for $i \neq j$, we have
$$
  \mathbb{E}[\| \sum_{k = 0}^{N-1}\hat{\zeta}(\Theta^{M}(t_{k}),
  \omega_{k}) \|^{2} ] = \sum_{i = 0}^{N-1}
  \mathbb{E}[\|\hat{\zeta}_{i} \|^{2}] \leq C_{0} NM^{-1} 
$$
for some $C_{0} < \infty$. Then there are constants $C_{0}, C_{1},
C_{2}, C_{3} < \infty$ such that
\begin{align*}
  \mathbb{E}[\| \Gamma^{M}(t_{N})\|^{2}] &\leq C_{1} \mathbb{E}[\|
  \sum_{k = 0}^{N-1} \hat{\zeta}_{k}\|^{2} ] + C_{1} \mathbb{E}[\|
  \sum_{k = 0}^{N-1}D(R \bar{\zeta})( \bar{\Theta}(t_{k}))
  \Gamma^{M}(t_{k})(t_{k+1} - t_{k})\|^{2}] + O(M^{-1})\\
&\leq C_{1} \mathbb{E}\left[ 
  \sum_{k = 0}^{N-1}D(R \bar{\zeta})(\bar{\Theta}(t_{k}))(t_{k+1} -
  t_{k})^{2}\right ]\cdot \mathbb{E}\left[ \sum_{k =
  0}^{N-1}\|\Gamma^{M}(t_{k})\|^{2}\right]\\
&+ C_{0}C_{1}\frac{N^{2}}{M^{2}} + O(M^{-1})\\
&\leq C_{0}C_{1}\frac{N^{2}}{M^{2}} + C_{0}C_{2}NM^{-2}
  \mathbb{E}\left[ \sum_{k =
  0}^{N-1}\|\Gamma^{M}(t_{k})\|^{2} \right]\\
&\leq C_{3}\frac{N^{2}}{M^{2}} + C_{3}\frac{N}{M^{2}} \sum_{k =
  0}^{N-1} \mathbb{E}\| \Gamma^{M}(t_{k})\|^{2} + O(M^{-2}) \,.
\end{align*}
Let $M_{N} = \max_{k\leq N}
\mathbb{E}[\|\Gamma^{M}(t_{k})\|^{2}]$. We have
$$
  M_{N} \leq C_{3}\frac{N^{2}}{M^{2}} + C_{3}\frac{N^{2}}{M^{2}} M_{N}
$$
Let $c_{0}$ be a sufficiently small number such that $C_{3}c_{0}^{2} <
1/2$, we have
$$
  M_{N} \leq C_{3}c_{0}^{2}(1 - C_{3}c_{0}^{2})^{-1}
$$
for $N = c_{0} M$. Hence
$$
  \mathbb{E}[\| \Gamma^{M}(t_{i}) - \Gamma^{M}(t)\|^{2}] \leq C_{3}c_{0}^{2}(1 - C_{3}c_{0}^{2})^{-1}
$$
for all $i < c_{0}M$. Since $c_{0} = O(1)$, we have $\|\Theta^{M}(s) -
\bar{\Theta}(s) \| = O(M^{-1/2})$ for all $0 < s < c_{0}$. This
estimate can be extended to all $s \in [0, T]$ because $T = O(1)$. 

\medskip

Now choose $h = 0.4 c_{0}$ and fix $N$ to be
$$
  N = M \int_{t}^{t + 2h} R( \bar{\Theta}(s)) \mathrm{d}s \,.
$$
Since $\|\Theta^{M}(s) - \bar{\Theta}(s) \| = O(M^{-1/2})$, $N$ is a
good approximation of total number of energy exchanges between $t$ and
$t + 2h$. In
addition, when $M$ is sufficiently large, the probabilities of $t_{N} -
t > c_{0}$ and $t +h > t_{N}$ all become negligible. Therefore, we have 
$$
  \mathbb{E}[\|\Gamma^{M}(t+h) - \Gamma^{M}(t)\|^{2}] \leq \max_{k\leq N}
\mathbb{E}[\|\Gamma^{M}(t_{k}) - \Gamma^{M}(t)\|^{2}] \leq 12.5 C_{3} h 
$$
for all sufficiently small $h > 0$. Since $c_{0}$ (hence $h$) can be
arbitrarily small, the proof is completed by applying
Theorem \ref{kurtztight}.
\end{proof}

\begin{lem}
\label{lem64}
for any test function $A \in C^{2}_{0}( \mathbb{R}^{N})$ and any
$t \in [0, T]$, we have
\begin{align}
  \label{eq6-3}
&\\\nonumber
&\lim_{M \rightarrow \infty} \mathbb{E}\left[  A(\Gamma^{M}(t)) -
  A(\Gamma^{M}(0)) - \int_{0}^{\infty} \nabla A(\Gamma^{M}(s)) \cdot D
  (R \bar{\zeta}(\bar{\Theta}(s)))\Gamma^{M}(s) \mathrm{d}s\right .\\\nonumber
& - \left . \frac{1}{2}\int_{0}^{\infty} \sum_{i = 1}^{N}\sum_{j = 1}^{N}
  \frac{\partial^{2}}{\partial \gamma_{i}\partial
  \gamma_{j}}A(\Gamma^{M}(s))R(\bar{\Theta}(s))\Sigma_{ij}(\bar{\Theta}(s))
  \mathrm{d}s \right ] = 0 \,,
  \nonumber
\end{align}
where 
$$
  \Sigma = \lim_{M \rightarrow \infty} M^{2} \mathbb{E}[
  \zeta\zeta^{T}] \,.
$$
\end{lem}
\begin{proof}
Let $\mathcal{L}^{M}$ be the infinitesimal generator of
$\Gamma^{M}(t)$. For any ${\bm \gamma} = (\gamma_{1}, \cdots, \gamma_{N})$,
let $\mathbf{E} = \bar{\Theta}(t) + M^{-1/2}{\bm \gamma}$ be the auxiliary variable,
we have
\begin{align*}
  \mathcal{L}^{M}A({\bm \gamma} ) = &Mf(T_{L}, E_{1}) \int_{[0,
  1]^{3}\times \mathbb{R}^{+}} (A(\gamma_{1} - \sqrt{M}(1-p)B_{1}E_{1} +
  \sqrt{M}pB_{2}X_{L}, \cdots, E_{N}) \\
& - A({\bm \gamma}) Q_{L}(p,B_{1},B_{2},X_{L})
  \mathrm{d}p \mathrm{d}B_{1} \mathrm{d}B_{2} \mathrm{d}X_{L}\\
+&Mf(E_{N}, T_{R}) \int_{[0,
  1]^{3}\times \mathbb{R}^{+}} (A(\gamma_{1}, \cdots, \gamma_{N} - \sqrt{M}(1-p)B_{1}E_{N} +
  \sqrt{M}pB_{2}X_{R}) \\
&- A({\bm \gamma})) Q_{R}(p,B_{1},B_{2},X_{L})
  \mathrm{d}p \mathrm{d}B_{1} \mathrm{d}B_{2} \mathrm{d}X_{R}\\
+& M\sum_{i = 1}^{N-1}f(E_{i}, E_{i+1}) \int_{[0, 1]^{3}}(A(\cdots,
   \gamma_{i} - \sqrt{M}(1-p)B_{1}E_{i} + \sqrt{M}pB_{2}E_{i+1}, \\
& \gamma_{i+1} - \sqrt{M}pB_{2}E_{i+1} +
  \sqrt{M}(1-p)B_{1}E_{i}, \cdots) - A( {\bm \gamma})) Q(p, B_{1},B_{2})
  \mathrm{d}p \mathrm{d}B_{1} \mathrm{d}B_{2} \\
-& \sqrt{M}\nabla A( {\bm \gamma})\cdot R( \bar{\Theta}(t))\bar{\zeta}(
  \bar{\Theta}(t)) \,,
\end{align*}
where joint probability density functions $Q_{L}, Q_{R}$, and $Q$ are
same as in the proof of Lemma \ref{martingale}. Similar as the case of
$\Theta^{M}(t)$, for any test function $A( \mathbf{E})$, 
$$
  A( \Gamma^{M}(t)) - A(\Gamma^{M}(0) - \int_{0}^{t} \mathcal{L}^{M}A(
  \Gamma^{M}(s)) \mathrm{d}s 
$$
is a martingale. Hence it is sufficient to show that $\mathcal{L}^{M}(
A(\Gamma^{M}(s))$ is a good approximation of 
$$
  \nabla A(\Gamma^{M}(s)) \cdot D
  (R \bar{\zeta}(\bar{\Theta}(s)))\Gamma^{M}(s) + \frac{1}{2}\sum_{i = 1}^{N}\sum_{j = 1}^{N}
  \frac{\partial^{2}}{\partial \gamma_{i}\partial
  \gamma_{j}}A(\Gamma^{M}(s))R(\bar{\Theta}(s))\Sigma_{ij}(\bar{\Theta}(s)) \,.
$$

Since Beta random variables with parameter $(1, M-1)$ is only
$O(M^{-1})$ small, a Taylor expansion of $A( \mathbf{E})$ gives 
\begin{align*}
&  \mathcal{L}^{M}A( {\bm \gamma} ) \\
=&  \frac{1}{2}M^{1/2}\sum_{i = 1}^{N}[ f(E_{i-1}, E_{i})(E_{i-1}
   -E_{i}) - f(E_{i}, E_{i+1})(E_{i+1} - E_{i})]A_{\gamma_{i}}\\
& -
   M^{1/2}\nabla A( {\bm \gamma} )\cdot R(
   \bar{\Theta}(t))\bar{\zeta}(\bar{\Theta}(t)) \\
&+ \frac{M^{2}}{2}f(T_{L}, E_{0})A_{\gamma_{1}\gamma_{1}}\int_{[0,
                                  1]^{3} \times \mathbb{R}^{+}}
                                  (-(1-p)B_{1}E_{1} + p B_{2}X_{L})^{2}Q_{L}(p,B_{1},B_{2},X_{L})
  \mathrm{d}p \mathrm{d}B_{1} \mathrm{d}B_{2} \mathrm{d}X_{L} \\
&+ \frac{M^{2}}{2}f(E_{N}, T_{R})A_{\gamma_{N}\gamma_{N}}\int_{[0,
                                  1]^{3} \times \mathbb{R}^{+}}
                                  (-(1-p)B_{1}E_{1} + p B_{2}X_{R})^{2}Q_{R}(p,B_{1},B_{2},X_{R})
  \mathrm{d}p \mathrm{d}B_{1} \mathrm{d}B_{2} \mathrm{d}X_{R}\\
&+ \sum_{i = 1}^{N-1} \left ( \frac{M^{2}}{2}f(E_{i}, E_{i+1})  \int_{[0,
  1]^{3}}A_{\gamma_{i}\gamma_{i}} (-(1-p)B_{1}E_{i} + pB_{2}E_{i+1})^{2} + A_{\gamma_{i+1}\gamma_{i+1}}(-p
  B_{2}E_{i+1} \right. \\& \left. + (1-p)B_{1}E_{i})^{2} 
 -2A_{\gamma_{i}\gamma_{i+1}}(-p
  B_{2}E_{i+1} + (1-p)B_{1}E_{i})^{2} Q(p, B_{1},B_{2}) \mathrm{d}p
  \mathrm{d}B_{1} \mathrm{d}B_{2} \right ) \\
&+ M\epsilon({\bm \gamma} )O(M^{-3/2}) \,,
\end{align*}
where the Lagrange reminder $\epsilon( {\bm \gamma} )$ depends on third
derivatives of $A$ and is bounded in a compact set. Note that when
${\bm \gamma} = \Gamma^{M}(s)$, the auxiliary variable $\mathbf{E}$ is
actually $\Theta^{M}(s)$. By the tightness of $\Gamma^{M}(t)$, term
$\Theta^{M}(s) - \bar{\Theta}(s)$ is $O(M^{-1/2})$ small. It remains
to calculate the coefficients of
$A_{\gamma_{i}\gamma_{j}}$, this is done in Lemma
\ref{covariance}. The coefficient of $A_{\gamma_{i}\gamma_{j}}$ is
indeed the $(i,j)$-th entry of $\mathbb{E}[\zeta(\mathbf{E},
\omega^{M})^{T}\zeta_{j}( \mathbf{E}, \omega^{M})]$. We denote the
rescaled limit of $M^{2} \mathbb{E}[
  \zeta( \mathbf{E}, \omega^{M}) \zeta^{T}( \mathbf{E}, \omega^{M})] $
  by $\Sigma$:
$$
  \Sigma( \mathbf{E}) = \lim_{M \rightarrow \infty} M^{2} \mathbb{E}[
  \zeta( \mathbf{E}, \omega^{M}) \zeta^{T}( \mathbf{E}, \omega^{M})] \,.
$$
It follows from the calculation in Lemma \ref{covariance} that the $(i,j)$-th entry of $\mathbb{E}[\zeta(\mathbf{E},
\omega^{M})^{T}\zeta_{j}( \mathbf{E}, \omega^{M})]$ is $O(M^{-1})$
close to $\Sigma_{ij}$. 

\medskip

Some further simplification and a Taylor expansion of $R \bar{\zeta}$ gives 
\begin{align*}
  \mathcal{L}^{M}A( \Gamma^{M}(s)) =& M^{1/2}\nabla A(\Gamma^{M}(s))
  \cdot R \bar{\zeta}(\Theta^{M}(s)) - M^{1/2}\nabla A(
  \Gamma^{M}(s)) \cdot \bar{R}\zeta( \bar{\Theta}(s)) \\
&+\frac{1}{2}\sum_{i = 1}^{N}\sum_{j = 1}^{N} A(\Gamma^{M}(s))R(
  \bar{\Theta}(s)) \Sigma_{ij}(\bar{\Theta}(s)) + M\epsilon({\bm
  \gamma} )O(M^{-3/2}) \\
=& \nabla A(\Gamma^{M}(s))\cdot D(R
   \bar{\zeta}(\bar{\Theta}(s)))\Gamma^{M}(s) + \frac{1}{2}\sum_{i = 1}^{N}\sum_{j = 1}^{N} A(\Gamma^{M}(s))R(
  \bar{\Theta}(s)) \Sigma_{ij}(\bar{\Theta}(s)) \\
&+ M\epsilon({\bm
  \gamma} )O(M^{-3/2})  +  O(M^{-1/2}) \,,
\end{align*}
where $\Sigma$ is the rescaled limit of the second moment of random vector $\zeta$:
$$
  \Sigma( \mathbf{E}) = \lim_{M \rightarrow \infty} M^{2} \mathbb{E}[
  \zeta( \mathbf{E}, \omega^{M}) \zeta^{T}( \mathbf{E}, \omega^{M})] \,.
$$

The proof is completed by letting $M \rightarrow 0$.
\end{proof}

It remains to calculate $\Sigma$, which follows immediately from the
second moment matrix of $X^{M}_{i}$. The following lemma follows from
straightforward calculations.

\begin{lem}
\label{covariance}
Let $\mathbf{E} = (E_{1}, \cdots, E_{N}) \in \mathbb{R}^{N}_{+}$. The
rescaled second moment matrix of
$\zeta(\mathbf{E}, \omega^{M})$ is
$$
  \Sigma = \lim_{M\rightarrow \infty} M^{2}\mathbb{E}[
  \zeta( \mathbf{E}, \omega^{M}) \zeta^{T}( \mathbf{E}, \omega^{M})]  = \frac{1}{R(\mathbf{E})}H(\mathbf{E})H(
  \mathbf{E})^{T} \,,
$$
where 
$$
  H(\mathbf{E}) = 
\begin{bmatrix}
V_{0}(T_{L}, E_{1}) & -V(E_{1}, E_{2}) &0&
0 &\cdots &\cdots \\
0&V(E_{1}, E_{2})& -V(E_{2}, E_{3})
&0&\cdots &\cdots\\
\vdots &\vdots &\vdots&\vdots&\vdots&\vdots\\
0&\cdots&\cdots&0&V(E_{N-1}, E_{N})&-V_{N}(E_{N}, T_{R})\\
\end{bmatrix} 
$$
is a $N \times (N+1)$ matrix, with
$$
  V(x_{1},x_{2}) = \sqrt{ f(x_{1},x_{2})\left (\frac{2}{3}x_{1}^{2} - \frac{1}{3}x_{1}x_{2} +
  \frac{2}{3}x_{2}^{2} \right )}\,,
$$
$$
  V_{0}(x_{1},x_{2}) = \sqrt{ f(x_{1},x_{2})\left (\frac{4}{3}x_{1}^{2} - \frac{1}{3}x_{1}x_{2} +
  \frac{2}{3}x_{2}^{2} \right )}\,,
$$
and
$$
  V_{N}(x_{1},x_{2}) = \sqrt{ f(x_{1},x_{2})\left (\frac{2}{3}x_{1}^{2} - \frac{1}{3}x_{1}x_{2} +
  \frac{4}{3}x_{2}^{2} \right )}\,,
$$
\end{lem}
\begin{proof}
Recall the definition of $X^{M} = \zeta( \mathbf{E}, \omega^{M})$,
where $\omega^{M} = (p_{1}, p_{2}, p_{3}, B_{1}, B_{2})$. 
Let
$$
  J_{k} = \left \{ 
\begin{array}[tb]{lll}
(1-p_{3})B_{1}E_{k}- p_{3}B_{2}E_{k+1} &
\mbox{ if } & 1 \leq k \leq N-1\\
-(1-p_{3})B_{1}T_{L}\log(1 - p_{2}) -
  p_{3}B_{2}E_{1} &\mbox{ if }&k = 0 \\
(1-p_{3})B_{1}E_{N} +
  p_{3}B_{2}T_{R}\log(1 - p_{2}) & \mbox{ if }& k =
                                                                  N 
\end{array}
\right. 
$$
be the energy flux from site $k$ to site $k+1$. For the sake of
simplicity denote $E_{0} = T_{L}$ and $E_{N+1} = T_{R}$. Then the $k$-th entry of $X^{M} $ is given by 
$$
  (X^{M})_{k} = -J_{k}\mathbf{1}_{\{\Delta = k\}} +
    J_{k-1}\mathbf{1}_{\{\Delta = k-1\}} \,,
$$
where $\Delta$ is a discrete random variable that takes value on $\{
0, 1, \cdots, N\}$ such that
$$
  \mathbb{P}[ \Delta = k] = \frac{f(E_{k}, E_{k+1})}{R( \mathbf{E}) } \,.
$$
Recall that $\Delta$ is chosen by $p_{1}$ that is independent of
$B_{1}, B_{2}, p_{2}, p_{3}$. Therefore, we have
\begin{displaymath}
  \mathbb{E}[(X^{M} )_{k}(X^{M})_{k}] 
=\frac{f(E_{k}, E_{k+1})}{R( \mathbf{E})} \mathbb{E}[(J_{k})^{2}] + \frac{f(E_{k-1},
   E_{k})}{R(\mathbf{E})} \mathbb{E}[(J_{k-1})^{2}]  \,,
\end{displaymath}
$$
  \mathbb{E}[(X^{M})_{k}(X^{M} )_{k+1}] = -\frac{f(E_{k}, E_{k+1})}{R( \mathbf{E})} \mathbb{E}[(J_{k})^{2}]  \,,
$$
and
$$
  \mathbb{E}[(X^{M})_{k}(X^{M})_{j}] = 0
$$
for all $j$ such that $|j - k| > 1$. Hence it remains to calculate $\mathbb{E}[(J_{k} )^{2}] $. For $k \neq 0, N$, we have
\begin{align*}
  &\mathbb{E}[(J_{k})^{2}]  = \mathbb{E}\left [ \left
  (E_{k}B_{1}(p_{3}-1) + E_{k+1}B_{2}p_{3}  \right )^{2} \right ]\\
=&E_{k}^{2} \mathbb{E}\left[ \left (B_{1}(p_{3}-1) 
   \right )^{2}\right
   ] + 2 E_{k}E_{k+1} \mathbb{E}\left[ B_{1}B_{2}(p_{3}-1)p_{3}\right
   ] + E_{k+1}^{2}\mathbb{E}\left[  (B_{1}(p_{3}-1)  )^{2}\right ]  \\
=& E_{k}^{2}\mathbb{E}[B_{1}^{2}]\mathbb{E}[(p_{3}-1)^{2}]  + 2
   E_{k}E_{k+1}\mathbb{E}[B_{1}] \mathbb{E}[B_{2}]
   \mathbb{E}[p_{3}(p_{3}-1)]  + E_{k+1}^{2}\mathbb{E}[B_{2}^{2}]
  \mathbb{E}[p_{3}^{2}] \\
=& E_{k}^{2} \frac{2}{3M(M+1)} - 2 E_{k}E_{k+1}\frac{1}{6M^{2}} +
   E_{k+1}^{2}\frac{2}{3M(M+1)} \,.
\end{align*}

Therefore, we have
\begin{displaymath}
  \mathbb{E}[J_{k})^{2}] = \frac{1}{M^{2}}\left [\frac{2}{3}E_{k}^{2}\frac{M}{M+1} -
   \frac{1}{3}E_{k}E_{k+1} + \frac{2}{3}E_{k+1}^{2}\frac{M}{M+1}
   \right ] \,.
\end{displaymath}
The case of $k = 0$ (resp. $k = N$) is identical, except the
expression becomes
$$
  \mathbb{E}[(J_{0})^{2}]  = \mathbb{E}\left [ \left (
  (E_{0}B_{1}(1 - p_{3})Z - E_{1}B_{2}p_{3}) \right )^{2} \right ] \,,
$$
(resp. 
$$
  \mathbb{E}[(J_{N})^{2}]  = \mathbb{E}\left [ \left (
  E_{N}B_{1}(p_{3}-1) + E_{N+1}B_{2}p_{3}Z \right )^{2} \right ] \,,
$$
)\\
where $Z$ is a standard exponential random variable that is
independent of other random variables. Since $\mathbb{E}[Z] = 1$ and
$\mathbb{E}[Z^{2}] = 2$, similar calculation shows that
\begin{align*}
    &\mathbb{E}[(J_{0})^{2}] = \frac{1}{M^{2}}\left [\frac{4}{3}E_{0}^{2}\frac{M}{M+1} -
   \frac{1}{3}E_{0}E_{1} + \frac{2}{3}E_{1}^{2}\frac{M}{M+1} \right
   ] 
\end{align*}
and 
\begin{align*}
    &\mathbb{E}[(J_{N})^{2}] = \frac{1}{M^{2}}\left [\frac{2}{3}E_{N}^{2}\frac{M}{M+1} -
   \frac{1}{3}E_{N}E_{N+1} + \frac{4}{3}E_{N}^{2}\frac{M}{M+1} \right
   ] \,.
\end{align*}

Therefore, we have
$$
 \mathbb{E}[\zeta( \mathbf{E}, \omega^{M})\zeta( \mathbf{E}, \omega^{M})^{T}] = 
\begin{bmatrix}
C_{0}+ C_{1} & -C_{1} & 0 &0&\cdots &0\\
-C_{1} & C_{1} + C_{2} & -C_{2} &0&\cdots &0\\
0&-C_{2}&C_{2}+C_{3}&-C_{3}&\cdots &0\\
\vdots&\vdots&\vdots&\vdots&\vdots&\vdots\\
0&\cdots& 0&-C_{N-2}&C_{N-2} + C_{N-1}&-C_{N-1}\\
0&\cdots&0&0&-C_{N-1}&C_{N-1} + C_{N} 
\end{bmatrix} \,,
$$
where
$$
  C_{0} =  \frac{f(E_{0}, E_{1})}{R( \mathbf{E})}\cdot\frac{1}{M^{2}}\left [\frac{4}{3}E_{0}^{2}\frac{M}{M+1} -
   \frac{1}{3}E_{0}E_{1} + \frac{2}{3}E_{1}^{2}\frac{M}{M+1} \right
   ]  \,,
$$
$$
  C_{N} = \frac{f(E_{N}, E_{N+1})}{R( \mathbf{E})}\cdot\frac{1}{M^{2}}\left [\frac{2}{3}E_{N}^{2}\frac{M}{M+1} -
   \frac{1}{3}E_{N}E_{N+1} + \frac{4}{3}E_{N}^{2}\frac{M}{M+1} \right
   ]  \,,
$$
and 
$$
  C_{k} = \frac{f(E_{k}, E_{k+1})}{R( \mathbf{E})}\cdot\frac{1}{M^{2}}\left [\frac{2}{3}E_{k}^{2}\frac{M}{M+1} -
   \frac{1}{3}E_{k}E_{k+1} + \frac{2}{3}E_{k+1}^{2}\frac{M}{M+1}
   \right ]
$$
for $k = 1, \cdots, N-1$. Now take the limit $M \rightarrow
\infty$. It is easy to see that
$$
  \lim_{M \rightarrow \infty} M^{2} \mathbb{E}[
  \zeta( \mathbf{E}, \omega^{M}) \zeta^{T}( \mathbf{E}, \omega^{M})]  = \frac{1}{R( \mathbf{E})}H H^{T} \,,
$$
where $H$ is given in the statement of the theorem. 
\end{proof}

It remains to show the uniqueness of solution to the martingale
problem given in Lemma \ref{lem64}. In general, the martingale problem
with respect to a differential operator has a unique solution if and
only if the corresponding stochastic differential equation has a
unique weak solution. See \cite{stroock2007multidimensional} for the full detail. 

\begin{lem}
\label{lem65}
The martingale problem given in Lemma \ref{lem64} has a unique solution.
\end{lem}
\begin{proof}
Notice that $L_{s}$ has timely dependent coefficients
$D(R(\bar{\Theta}(t))\bar{\zeta}(t))$ and
$R(\bar{\Theta}(t))\Sigma_{ij}(t)$ that are
uniformly bounded. Hence there exists a constant $C<\infty$ such that
$$
  |(D(R(\bar{\Theta}(t))\bar{\zeta}(t)) \mathbf{E} \cdot \mathbf{E})|
  \leq C (\| \mathbf{E} \|) + 1\,.
$$
The lemma then follows from Theorem \ref{book1022}. 
\end{proof}

\begin{proof}[Proof of Theorem \ref{thm61}]
Lemma \ref{gammatight} implies that
$\{\Gamma^{M}(t), t \in [0 ,T ]\}$ is tight. Then it follows from Lemma \ref{lem64} that any sequential limit of
$\{\Gamma^{M}(t), t \in [0 ,T ]\}$ solves the martingale problem
described by equation \eqref{eq6-3}. Finally, it follows from Lemma \ref{lem65}
that the martingale problem given by equation \eqref{eq6-3} has a
unique solution. Therefore, the unique limit of $\Gamma^{M}(t)$,
denoted by $\bar{\Gamma}(t)$, has a time-dependent generator
\begin{equation}
\label{generator}
  (\mathcal{L}_{t}A)( \mathbf{E}) = \left (D(R(\bar{\Theta}(t)) \bar{\zeta}(
  \bar{\Theta}(t) )) \mathbf{E} \cdot \nabla A( \mathbf{E}) \right ) +
\frac{1}{2}\sum_{i = 1}^{N}\sum_{j = 1}^{N}
R(\bar{\Theta}(t))\Sigma_{ij}(
\bar{\Theta}(t))\frac{\partial^{2}}{\partial E_{i} \partial E_{j}}A(
\mathbf{E}) \,,
\end{equation}
where $\Sigma$ is calculated in Lemma \ref{covariance}. Therefore, we
have
$$
  R(\bar{\Theta}(t))\Sigma_{ij}(
\bar{\Theta}(t)) = H( \bar{\Theta}(t))H( \bar{\Theta}(t))
$$
for the matrix-valued function given in Lemma \ref{covariance}. Hence
$\bar{\Gamma}(t)$ satisfies the stochastic differential equation
\eqref{CLT}. This completes the proof.

\end{proof}

Theorem \ref{LLN} and Theorem \ref{thm61} implies that
$$
  \Theta^{M}(t) \approx \bar{\Theta}(t) + M^{-1/2}\Gamma(t) \,.
$$
Some calculation in the following lemma gives the error bound of this
approximation.
\begin{lem}
\label{lem68}
For each $t \in [0, T]$, we have
$$
  \mathbb{E}[ \| \Theta^{M}(t) - (\bar{\Theta}(t) +
  M^{-1/2}\Gamma(t)) \| ] \leq C M^{- 1}  \,,
$$
where $C < \infty$ is a constant independent
of $t$ and $M$.
\end{lem}
\begin{proof}
Recall that 
$$
  \Theta^{M}(t) = \bar{\Theta}(t) + M^{-1/2}\Gamma^{M}(t) \,.
$$
Hence it is sufficient to estimate the difference $\|\Gamma^{M}(t) -
\Gamma(t)\|$ for $t \in [0, T]$.

The proof of Lemma \ref{lem64} gives a bound 
\begin{equation}
\label{gammaM}
  \mathbb{E}\left [ A(\Gamma^{M}(t)) - A(\Gamma^{M}(0)) - \int_{0}^{t}
  \mathcal{L}_{s}A(\Gamma^{M}(s)) \mathrm{d}s\right ] = O(M^{-1/2}) \,,
\end{equation}
where $\mathcal{L}_{s}$ is the timely dependent generator given in
equation \eqref{generator}. 

In addition, $\Gamma(t)$ solves the martingale problem means
\begin{equation}
\label{gammaLimit}
  \mathbb{E}\left [ A(\Gamma(t)) - A(\Gamma(0)) - \int_{0}^{t}
  \mathcal{L}_{s}A(\Gamma(s)) \mathrm{d}s\right ]  = 0 \,.
\end{equation}
Let $A_{\mathbf{v}}$ be a smooth test function such that $A_{\mathbf{v}}(\mathbf{x}) =
\mathbf{v}^{T} \mathbf{x}$ for all $\|\mathbf{x}\| < M^{\epsilon}$, where
$\mathbf{v} \in \mathbb{R}^{N}$ is a unit vector. Then by Lemma \ref{unibound}, the probability
that $\Gamma(t)$ travels out side of the $M^{\epsilon}$-ball is
negligibly small. And the probability that $\| \Gamma^{M}(t) -
\Gamma(t)\| \geq 1$ is at most $O(M^{- 1/2})$ because of
equations \eqref{gammaM} and \eqref{gammaLimit}. Therefore, terms
$\mathcal{L}_{s}A_{\mathbf{v}}(\Gamma(s))$ and $\mathcal{L}_{s}A_{\mathbf{v}}(\Gamma^{M}(s))$ becomes two identical constant vectors plus $O(M^{ - 1/2})$ terms. Hence for any $t \in [0,
T]$ and unit vector $\mathbf{v}$, we have
$$
  \mathbb{E}[ \mathbf{v}^{T} \cdot (\Gamma^{M}(t) - \Gamma(t))] \leq
  C_{0}M^{- 1/2}
$$
for some $C_{0} < \infty$. This implies 
$$
  \mathbb{E}[ \| \Gamma^{M}(t)  - \Gamma(t) \| ] \leq C_{0} M^{- 1/2} \,.
$$
This argument applies for any $t \in [0, T]$. Then it follows from the
definition of $\Gamma^{M}(t)$ that
$$
  \mathbb{E}[\|  \Theta^{M}(t) - (\bar{\Theta}(t) +
  M^{-1/2}\Gamma(t)) \| ]\leq C M^{- 1} \,,
$$
for any $t \in [0, T]$, where $C = C_{0} < \infty$. This completes the
proof. 
\end{proof}

Finally, the following Proposition shows that $\bar{\Theta}(t) +
M^{-1/2}\Gamma(t) $ is approximated by a stochastic differential equation.

\begin{pro}
\label{pro69}
Let $Z_{t}$ be a stochastic differential equation satisfying 
\begin{equation}
\label{global}
  \mathrm{d}Z_{t} = R(Z_{t})\bar{\zeta}(Z_{t}) \mathrm{d}t + M^{-1/2}
  H(Z_{t})\mathrm{d}W_{t} \,.
\end{equation}
Then we have
$$
  Z_{t} = \bar{\Theta}(t) + M^{-1/2}\Gamma(t) + R(t) \,,
$$
where 
\begin{equation}
  \mathbb{E}[ \|R(t) \| ] < C M^{-1}
\end{equation}
for some constant $C < \infty$ that is independent of $t \in [0, T]$
and $M$. 
\end{pro}
\begin{proof}
This proposition follows from Chapter 2 Theorem 2.1 of \cite{FW}. The
calculation in the proof of Chapter 2 Theorem 2.1 of \cite{FW} implies that 
$$
  \| R(t) \| \leq M^{-1}C_{0}( \max_{0 \leq s \leq t} \| W_{t} \| )^{2} \,,
$$
where $W_{t}$ is the $(N+1)$-dimensional Brownian motion. Hence there
exists a constant $C$ such that 
$$
  \mathbb{E}[ \| R(t) \|] < CM^{-1}
$$
for all $t \in [0, T]$.
\end{proof}

Proposition \ref{pro69} and Lemma \ref{lem68} implies the following
corollary immediately.
\begin{cor}
\label{cor610}
Let $Z_{t}$ be a stochastic differential equation given by equation
\eqref{global}. Then for any $0 < \epsilon \ll 1$, we have
\begin{equation}
  \mathbb{E}[ \|Z_{t} - \Theta^{M}(t) \| ]< C M^{-1} 
\end{equation}
for some constant $C < \infty$ that is independent of $t \in [0, T]$
and $M$. 
\end{cor}

Equation \ref{global} is called the mesoscopic limit equation. We will
work on macroscopic thermodynamic properties of this equation in our
subsequent work. 

\section{Conclusion}
In this paper we continue to work on the stochastic energy exchange
model for heat conduction in gas. This stochastic energy exchange
model is an approximation of a billiards-like deterministic heat
conduction model, which is unfortunately not mathematically
tractable. In this paper, we consider the mesoscopic limit, which
means the number of particles within a cell, denoted by $M$, approaches to
infinity. The time of the stochastic energy exchange model is then
rescaled, such that the mean heat flux is independent
of $M$. 

We use martingale problem to prove that as $M\rightarrow \infty$, the
trajectory of the stochastic energy exchange model converges to the
solution of a nonlinear discrete heat equation almost surely. Fourier's law holds for the equilibrium of this nonlinear
discrete heat equation. In addition, a similar martingale
problem gives us the central limit theorem, which means the rescaled
difference between the stochastic energy exchange
model and that of the discrete heat equation converges to a stochastic
differential equation as $M \rightarrow \infty$. Therefore, for large
but finite $M$, trajectories of the stochastic energy exchange model
is approximated by a stochastic differential equation with small
random perturbation, which is called the {\it mesoscopic limit
  equation}. 

An important observation of the invariant probability measure of the mesoscopic limit
equation \eqref{global}, denoted by $\pi_{Z}$, is a close
approximation of that of the original
stochastic energy exchange process $\Theta^{M}(t)$, denoted by
$\pi_{\Theta}$. This is because we have good control of the finite
time error between the law of $Z_{t}$ and that of $\Theta^{M}(t)$ in
Corollary \ref{cor610}. If we also know the speed of convergence to
$\pi_{Z}$ for $Z_{t}$, then the distance between $\pi_{Z}$ and
$\pi_{\Theta}$ can be bounded. (See \cite{johndrow2017error,
  dobson2019using}.) The deterministic part of $Z_{t}$
admits a stable equilibrium (Lemma 5.7), so it is not hard to show that the law of
$Z_{t}$ converges to $\pi_{Z}$ exponentially fast. In addition, the
probability density function of $\pi_{Z}$, denoted by $\rho_{Z}$, can
be approximated by an WKB expansion
$$
  \rho_{Z}(\mathbf{E}) \approx \frac{1}{K}e^{-( \mathbf{E} -
    \mathbf{E}^{*})^{T}\mathbf{S}( \mathbf{E} - \mathbf{E}^{*})/(2M)} \,,
$$
where $\mathbf{E}^{*}$ is the stable equilibrium given in Lemma 5.6,
and $\mathbf{S}$ solves the Lyapunov equation
$$
  {\bf S}{\bm J}( \mathbf{E}^{*})^{T} + {\bm J}( \mathbf{E}^{*}){\bf S} + H( \mathbf{E}^{*})H( \mathbf{E}^{*})^{T} = 0 
$$
for the Jacobian ${\bm J}$ at $\mathbf{E}^{*}$. In other words
$\pi_{Z}$ is approximated by a Gaussian
distribution. The covariance matrix $\mathbf{S}$ of this Gaussian distribution is
the solution of a Lyapunov equation. Further calculation
shows that $\mathbf{S}$ is an $O(M^{-1})$
perturbation of a diagonal matrix. Therefore, many interesting
properties, including the long range correlation, entropy
production rate, and fluctuation theorem, can be proved for both the
global equation $Z_{t}$ and the original stochastic energy exchange
process $\Theta^{M}(t)$. We decide to 
put results about thermodynamic properties of $\Theta^{M}(t)$ and
$Z_{t}$ into our 
subsequent paper, as techniques used for these results are very
different from those in the present paper.

\bibliography{myref}{}

\providecommand{\bysame}{\leavevmode\hbox to3em{\hrulefill}\thinspace}
\providecommand{\MR}{\relax\ifhmode\unskip\space\fi MR }
\providecommand{\MRhref}[2]{%
  \href{http://www.ams.org/mathscinet-getitem?mr=#1}{#2}
}
\providecommand{\href}[2]{#2}
\begin{thebibliography}{10}

\bibitem{anderson2011continuous}
David~F Anderson and Thomas~G Kurtz, \emph{Continuous time markov chain models
  for chemical reaction networks}, Design and analysis of biomolecular
  circuits, Springer, 2011, pp.~3--42.

\bibitem{bernardin2005fourier}
C{\'e}dric Bernardin and Stefano Olla, \emph{Fourier’s law for a microscopic
  model of heat conduction}, Journal of Statistical Physics \textbf{121}
  (2005), no.~3, 271--289.

\bibitem{bonetto2000fourier}
F.~Bonetto, J.L. Lebowitz, and L.~Rey-Bellet, \emph{Fourier's law: a challenge
  to theorists}, Mathematical physics 2000 (2000), 128--150.

\bibitem{bonetto2009heat}
Federico Bonetto, Joel~L Lebowitz, Jani Lukkarinen, and Stefano Olla,
  \emph{Heat conduction and entropy production in anharmonic crystals with
  self-consistent stochastic reservoirs}, Journal of Statistical Physics
  \textbf{134} (2009), no.~5, 1097--1119.

\bibitem{bunimovich1992ergodic}
Leonid Bunimovich, Carlangelo Liverani, Alessandro Pellegrinotti, and Yurii
  Suhov, \emph{Ergodic systems of n balls in a billiard table}, Communications
  in mathematical physics \textbf{146} (1992), no.~2, 357--396.

\bibitem{caprini2017fourier}
Lorenzo Caprini, Luca Cerino, Alessandro Sarracino, and Angelo Vulpiani,
  \emph{Fourier’s law in a generalized piston model}, Entropy \textbf{19}
  (2017), no.~7, 350.

\bibitem{de2015martingale}
Jacopo De~Simoi and Carlangelo Liverani, \emph{The martingale approach after
  varadhan and dolgopyat}, Hyperbolic dynamics, fluctuations and large
  deviations \textbf{89} (2015), 311--339.

\bibitem{dobson2019using}
Matthew Dobson, Jiayu Zhai, and Yao Li, \emph{Using coupling methods to
  estimate sample quality for stochastic differential equations}, arXiv
  preprint arXiv:1912.10339 (2019).

\bibitem{dolgopyat2011energy}
Dmitry Dolgopyat and Carlangelo Liverani, \emph{Energy transfer in a fast-slow
  hamiltonian system}, Communications in Mathematical Physics \textbf{308}
  (2011), no.~1, 201--225.

\bibitem{eckmann2000non}
J-P Eckmann and Martin Hairer, \emph{Non-equilibrium statistical mechanics of
  strongly anharmonic chains of oscillators}, Communications in Mathematical
  Physics \textbf{212} (2000), no.~1, 105--164.

\bibitem{eckmann1999entropy}
Jean-Pierre Eckmann, Claude-Alain Pillet, and Luc Rey-Bellet, \emph{Entropy
  production in nonlinear, thermally driven hamiltonian systems}, Journal of
  statistical physics \textbf{95} (1999), no.~1-2, 305--331.

\bibitem{eckmann1999non}
\bysame, \emph{Non-equilibrium statistical mechanics of anharmonic chains
  coupled to two heat baths at different temperatures}, Communications in
  Mathematical Physics \textbf{201} (1999), no.~3, 657--697.

\bibitem{ethier2009markov}
Stewart~N Ethier and Thomas~G Kurtz, \emph{Markov processes: characterization
  and convergence}, vol. 282, John Wiley \& Sons, 2009.

\bibitem{fourier1822theorie}
Joseph Fourier, \emph{Theorie analytique de la chaleur, par m. fourier}, Chez
  Firmin Didot, p{\`e}re et fils, 1822.

\bibitem{FW}
Mark Freidlin and Alexander~D Wentzell, \emph{Random perturbations of dynamical
  systems}, vol. 260, Springer, 2012.

\bibitem{gaspard2008heat2}
Pierre Gaspard and Thomas Gilbert, \emph{Heat conduction and fourier's law in a
  class of many particle dispersing billiards}, New Journal of Physics
  \textbf{10} (2008), no.~10, 103004.

\bibitem{gaspard2008heat}
\bysame, \emph{Heat conduction and fourier’s law by consecutive local mixing
  and thermalization}, Physical review letters \textbf{101} (2008), no.~2,
  020601.

\bibitem{gaspard2008derivation}
\bysame, \emph{On the derivation of fourier's law in stochastic energy exchange
  systems}, Journal of Statistical Mechanics: Theory and Experiment
  \textbf{2008} (2008), no.~11, P11021.

\bibitem{grigo2012mixing}
A.~Grigo, K.~Khanin, and D.~Szasz, \emph{Mixing rates of particle systems with
  energy exchange}, Nonlinearity \textbf{25} (2012), no.~8, 2349.

\bibitem{jennings1988mean}
SG~Jennings, \emph{The mean free path in air}, Journal of Aerosol Science
  \textbf{19} (1988), no.~2, 159--166.

\bibitem{johndrow2017error}
James~E Johndrow and Jonathan~C Mattingly, \emph{Error bounds for
  approximations of markov chains used in bayesian sampling}, arXiv preprint
  arXiv:1711.05382 (2017).

\bibitem{kipnis1982heat}
C.~Kipnis, C.~Marchioro, and E.~Presutti, \emph{Heat flow in an exactly
  solvable model}, Journal of Statistical Physics \textbf{27} (1982), no.~1,
  65--74.

\bibitem{lepri2003thermal}
Stefano Lepri, Roberto Livi, and Antonio Politi, \emph{Thermal conduction in
  classical low-dimensional lattices}, Physics reports \textbf{377} (2003),
  no.~1, 1--80.

\bibitem{li2018polynomial}
Yao Li, \emph{On the polynomial convergence rate to nonequilibrium steady
  states}, The Annals of Applied Probability \textbf{28} (2018), no.~6,
  3765--3812.

\bibitem{li2018billiards}
Yao Li and Lingchen Bu, \emph{From billiards to thermodynamic laws: Stochastic
  energy exchange model}, Chaos: An Interdisciplinary Journal of Nonlinear
  Science \textbf{28} (2018), no.~9, 093105.

\bibitem{li2013existence}
Yao Li and Lai-Sang Young, \emph{Existence of nonequilibrium steady state for a
  simple model of heat conduction}, Journal of Statistical Physics \textbf{152}
  (2013), no.~6, 1170--1193.

\bibitem{li2014nonequilibrium}
\bysame, \emph{Nonequilibrium steady states for a class of particle systems},
  Nonlinearity \textbf{27} (2014), no.~3, 607.

\bibitem{liverani2011toward}
C.~Liverani and S.~Olla, \emph{Toward the fourier law for a weakly interacting
  anharmonic crystal}, Journal of the American Mathematical Society \textbf{25}
  (2011), 555--583.

\bibitem{rey2001exponential}
Luc Rey-Bellet and L~Thomas, \emph{Exponential convergence to non-equilibrium
  stationary states in classical statistical mechanics}, Communications in
  mathematical physics \textbf{255} (2001), no.~2, 305--329.

\bibitem{rey2000asymptotic}
Luc Rey-Bellet and Lawrence~E Thomas, \emph{Asymptotic behavior of thermal
  nonequilibrium steady states for a driven chain of anharmonic oscillators},
  Communications in Mathematical Physics \textbf{215} (2000), no.~1, 1--24.

\bibitem{rey2002fluctuations}
\bysame, \emph{Fluctuations of the entropy production in anharmonic chains},
  Annales Henri Poincare, vol.~3, Springer, 2002, pp.~483--502.

\bibitem{ruelle1996positivity}
David Ruelle, \emph{Positivity of entropy production in nonequilibrium
  statistical mechanics}, Journal of Statistical Physics \textbf{85} (1996),
  no.~1, 1--23.

\bibitem{ruelle1997entropy}
\bysame, \emph{Entropy production in nonequilibrium statistical mechanics},
  Communications in Mathematical Physics \textbf{189} (1997), no.~2, 365--371.

\bibitem{sasada2015spectral}
Makiko Sasada et~al., \emph{Spectral gap for stochastic energy exchange model
  with nonuniformly positive rate function}, The Annals of Probability
  \textbf{43} (2015), no.~4, 1663--1711.

\bibitem{simanyi2003proof}
N{\'a}ndor Sim{\'a}nyi, \emph{Proof of the boltzmann-sinai ergodic hypothesis
  for typical hard disk systems}, Inventiones Mathematicae \textbf{154} (2003),
  no.~1, 123--178.

\bibitem{simanyi1999hard}
N{\'a}ndor Sim{\'a}nyi and Domokos Sz{\'a}sz, \emph{Hard ball systems are
  completely hyperbolic}, Annals of Mathematics \textbf{149} (1999), 35--96.

\bibitem{spohn1983long}
Herbert Spohn, \emph{Long range correlations for stochastic lattice gases in a
  non-equilibrium steady state}, Journal of Physics A: Mathematical and General
  \textbf{16} (1983), no.~18, 4275.

\bibitem{stroock2007multidimensional}
Daniel~W Stroock and SR~Srinivasa Varadhan, \emph{Multidimensional diffusion
  processes}, Springer, 2007.

\end{thebibliography}
\bibliographystyle{amsplain}
\end{document}